\documentclass[10pt,journal,comsoc]{IEEEtran}  

\usepackage{cite}
\usepackage{amsmath,amssymb,amsfonts}
\usepackage{graphicx}
\usepackage{graphics}
\usepackage{subfigure}
\usepackage{algorithm}
\usepackage{algorithmic}
\usepackage{multirow}
\usepackage{array}
\usepackage{textcomp}
\usepackage{xcolor}
\usepackage{multicol}
\usepackage{makecell}
\usepackage{bm}

\usepackage{amsthm}

\newtheorem{definition}{Definition}

\newtheorem{thm}{Theorem}

\def\BibTeX{{\rm B\kern-.05em{\sc i\kern-.025em b}\kern-.08em
		T\kern-.1667em\lower.7ex\hbox{E}\kern-.125emX}}

\begin{document}

\title{Adaptive and Fair Deployment Approach to Balance Offload Traffic in Multi-UAV Cellular Networks}


\author{Chuan-Chi~Lai,~\IEEEmembership{Member,~IEEE}, Bhola,~\IEEEmembership{Student~Member,~IEEE}, Ang-Hsun~Tsai,~\IEEEmembership{Member,~IEEE}, and~Li-Chun~Wang,~\IEEEmembership{Fellow,~IEEE}%
	\IEEEcompsocitemizethanks{
		\IEEEcompsocthanksitem Copyright (c) 2015 IEEE. Personal use of this material is permitted. However, permission to use this material for any other purposes must be obtained from the IEEE by sending a request to pubs-permissions@ieee.org.
		\IEEEcompsocthanksitem This work has been partially funded by the Ministry of Science and Technology under the Grants MOST 110-2222-E-035-004-MY2, MOST 110-2634-F-A49-006-, and MOST 110-2221-E-A49-039-MY3, Taiwan. This work was also financially supported by the Center for Open Intelligent Connectivity from The Featured Areas Research Center Program within the framework of the Higher Education Sprout Project by the Ministry of Education (MOE) in Taiwan. This work was supported by the Higher Education Sprout Project of the National Yang Ming Chiao Tung University and Ministry of Education (MOE), Taiwan. \emph{(Corresponding author: Li-Chun Wang.)}
		\IEEEcompsocthanksitem C.-C. Lai is with the Department of Information Engineering and Computer Science, Feng Chia University, Taichung 40724, Taiwan.
		\IEEEcompsocthanksitem Bhola is with the Electrical Engineering and Computer Science International Graduate Program, National Yang Ming Chiao Tung University, Hsinchu 30010, Taiwan. 
		\IEEEcompsocthanksitem A.-H. Tsai is with the Department of Communications Engineering, Feng Chia University, Taichung 40724, Taiwan.
		\IEEEcompsocthanksitem L.-C. Wang is with the Department of Electrical and Computer Engineering, National Yang Ming Chiao Tung University, Hsinchu 30010, Taiwan. 
	}
}

\markboth{Preprint for IEEE Transactions on Vehicular Technology}
{}

\maketitle

\begin{abstract}
	Unmanned aerial vehicle-aided communication (UAB-BS) is a promising solution to establish rapid wireless connectivity in sudden/temporary crowded events because of its more flexibility and mobility features than conventional ground base station (GBS). Because of these benefits, UAV-BSs can easily be deployed at high altitudes to provide more line of sight (LoS) links than GBS. Therefore, users on the ground can obtain more reliable wireless channels. In practice, the mobile nature of the ground user can create uneven user density at different times and spaces. This phenomenon leads to unbalanced user associations among UAV-BSs and may cause frequent UAV-BS overload. We propose a three-dimensional adaptive and fair deployment approach to solve this problem. The proposed approach can jointly optimize the altitude and transmission power of UAV-BS to offload the traffic from overloaded UAV-BSs. The simulation results show that the network performance improves by 37.71\% in total capacity, 37.48\% in total energy efficiency and 16.12\% in the Jain fairness index compared to the straightforward greedy approach.
\end{abstract}

\begin{IEEEkeywords}
	UAV base station, traffic offload, Jain fairness index, energy efficiency
\end{IEEEkeywords}

\section{Introduction}
\label{sec:introduction}
\IEEEPARstart{T}{he} unmanned aerial vehicle base station (UAV-BS) has recently attracted significant attention. It has many unexplored applications and could be a promising solution for current and future wireless communication systems. UAV-BS has some significant advantages over the terrestrial/ground base station (GBS). For example, when GBS malfunctions or is unavailable in disaster and hotspot areas, UAV-BS networks can rapidly deploy and establish emergency communications~\cite{9295376}\cite{coelho2022traffic}. For example, China recently deployed a drone-based wireless access point for emergency communications and damage assessment in areas affected by the floods~\cite{R1}. The specialist drone, Wing Loong 2H, is used to fly from the south side to the central location of Henan province, China, which was disabled by power failure and wireless network outages. The drone provided 5 hours of network service to a flooded hospital where terrestrial communications could not be restored.

According to this success story, UAV-BS has become a key carrier to provide beyond 5G networks (B5G). Unlike traditional GBSs, the UAV-BS networks are adaptive in multiple parameters, such as altitude, and transmission power, three-dimensional (3D) location~\cite{R13}~\cite{9448994}. The deployment of UAV-BS is very flexible under any unrealistic conditions or time constraints on the ground~\cite{R18}. Benefiting from the above advantages, UAV-BS has a higher probability of providing line-of-sight (LoS) signals than GBS, which guarantees better quality of service (QoS) for ground users~\cite{9625734}~\cite{li2022network}. 

Although a single UAV base station shows advantages in improving wireless network performance, this is still limited by size, weight, power consumption (SWaP), and limited computing power, which directly affects the maximum flight altitude, communication coverage, service endurance~\cite{R18}, and capacity~\cite{R4}.
Thus, the service capacity (maximum number of associated users) of each UAV-BS is limited and may not guarantee availability during the entire mission. A swarm UAV-BS network can provide a longer transmission range, complete missions faster at a lower cost, and achieve more balanced management of traffic offloading than a single UAV-BS network~\cite{R15}. Therefore, we conclude that the swarm of the UAV-BS network is suitable for many applications, such as in the temporary or sudden surge of bursty communication scenarios, like disaster search and rescue operations~\cite{R3}~\cite{R14}, live concerts, and traffic overload~\cite{R2}. Thus, we are motivated to use a swarm of UAV-BSs in this work.

\begin{figure*}
	\centering
	\includegraphics[width=.85\textwidth]{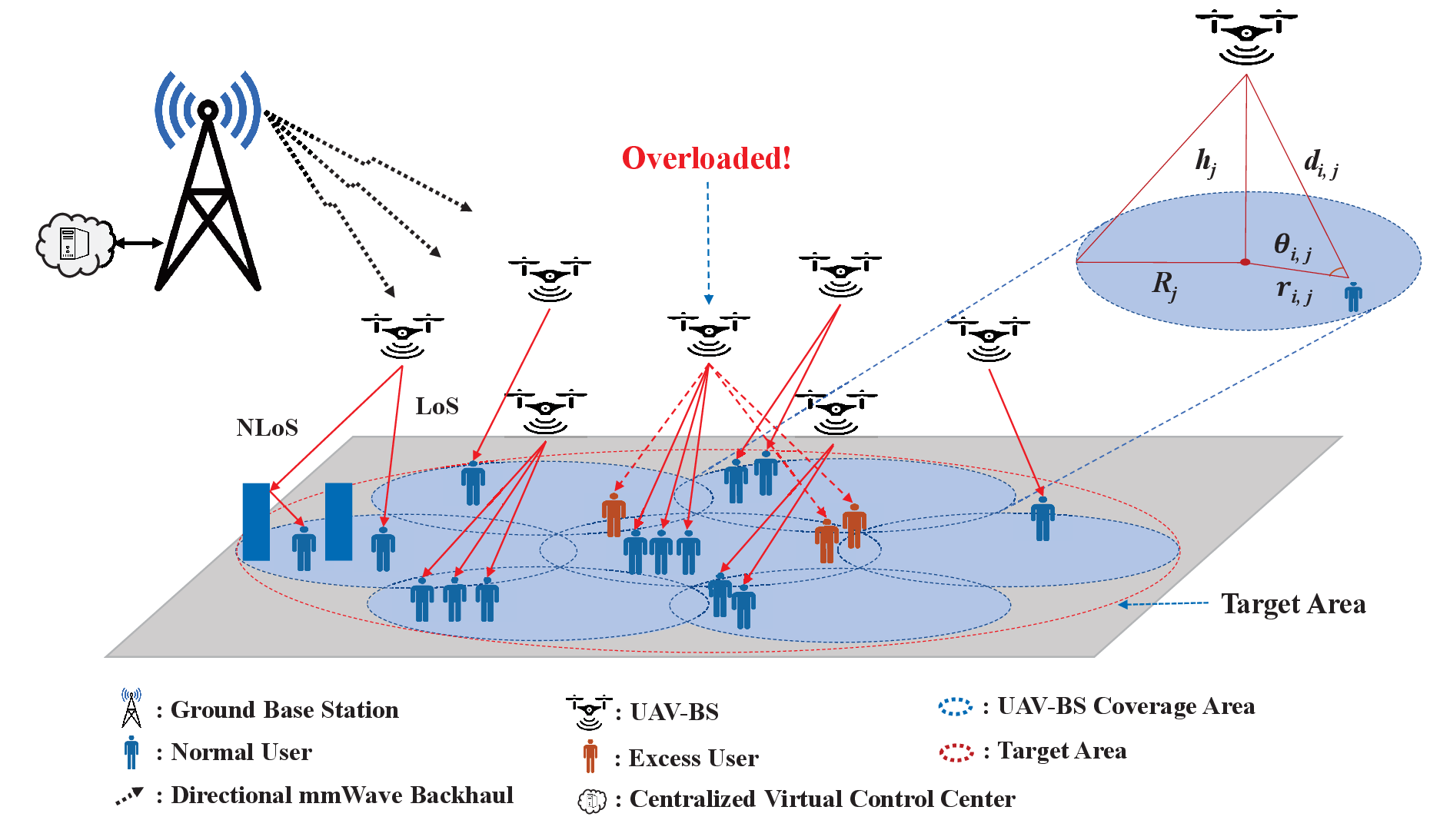}
	\vspace{-5pt}
	\caption{Illustration of the considered system model.}
	\label{fig:fig1.pdf}
\end{figure*}

Additionally, user mobility may cause uneven user density at different times and locations, resulting in frequent overloading of UAV-BSs. The number of available of UAV-BSs and user association capacity of a UAV-BS are limited. Under the above constraints, the basic requirement of QoS is that the uneven distribution of users should not affect ongoing user calls. If users are unevenly distributed, QoS will be degraded~\cite{6736745}, users will not be able to obtain fair Internet access and meet the latency requirements, and even UAV-BSs and users will consume much energy during this period~\cite{R25}.

To solve the traffic overload problem of the considered multi-UAV cellular network, shown in Fig.~\ref{fig:fig1.pdf}, we propose an adaptive and fair deployment (AFD) approach to dynamically control multiple UAV-BSs to provide fair traffic offloading opportunities for ground users. We can improve system capacity, total energy efficiency, and fair user association. The proposed AFD approach allows overloaded UAV-BS to offload the excess load/association to their neighboring available UAV-BSs in response to the nearby available UAV-BSs can reposition and serve the excess user from the overloaded cell. 

In fact, the traffic overload problem is also a popular topic in terrestrial cellular networks. A well-known solution is the cell breathing algorithm~\cite{685159}~\cite{DEMIRCI2018140}. Since terrestrial base stations (or access points) are usually deployed at fixed locations and altitudes, the cell breathing algorithm can only controls the transmit power of all base stations to achieve traffic load balancing. However, increasing the transmit power of the base station usually leads to severe inter-cell interference, resulting in poor energy efficiency performance~\cite{7422407}~\cite{8369148}~\cite{8411547}.
As mentioned earlier, UAV-BS has one more degree of freedom, flexibility in 3D position (especially height), which enables UAV-BS to provide UE with higher LoS link probability without increasing transmit power/inter-cell interference, thus improving network performance~\cite{8642333}~\cite{9454276}~\cite{R26}.

Hence, the contributions of this work are summarized as follows:
\begin{itemize}
	\item We identify a UAV-BS overload problem caused by the uneven distribution of ground users in a multi-UAV cellular network. 
	\item To solve the UAV-BS overload problem, we formulate an optimization problem to maximize the total energy efficiency of the multi-UAV-assisted cellular network by jointly optimizing the altitude, transmission power, and fair user association.
	\item Then, we propose an AFD approach that enables neighboring UAV-BSs to share overloaded traffic and jointly optimize the altitude and transmission power of each UAV-BS to meet a predefined fairness requirement (e.g., Jain fairness index (JFI)~\cite{R19}). 
	\item The simulation results show that the proposed AFD is the best approach. Compared to the straightforward greedy approach, the proposed AFD can improve the system performance in total capacity, total energy efficiency, and JFI value by 37.71\%, 37.48\%, and 16.12\%, respectively.
\end{itemize}

Section~\ref{sec:Related work} discusses the related work. Section~\ref{sec:System Model} introduces the system model. Section~\ref{sec: Problem Formulation} discusses the problem formulation. Section~\ref{sec:AFD} explains the proposed AFD approach. The simulation results and comparison summary is presented in Section~\ref{sec: Simulation Results and Performance Analysis}. Finally, Section VII presents the conclusion.

\begin{table*}[!ht]
	\caption{Comparative summary of Related Works and the Proposed Method}
	\label{Comparisons of Related Works and the Proposed Method}
	\centering
	\resizebox{\textwidth}{!}{%
		\begin{tabular}{|c|p{4.6cm}|c|c|c|c|c|c|}
			\hline
			\textbf{Method} &
			\multicolumn{1}{c|} {\textbf{Objective}} &
			\textbf{\begin{tabular}[c]{@{}c@{}}Number of UAVs\end{tabular}} &
			\textbf{\begin{tabular}[c]{@{}c@{}}User Association \\ Capacity\end{tabular}} &
			\textbf{\begin{tabular}[c]{@{}c@{}}User \\ Mobility\end{tabular}} &
			\textbf{\begin{tabular}[c]{@{}c@{}}Transmit Power \\ Control\end{tabular}} &
			\textbf{\begin{tabular}[c]{@{}c@{}}Fairness \end{tabular}} &
			\textbf{\begin{tabular}[c]{@{}c@{}}Traffic \\ Offload\end{tabular}} \\ \hline\hline
			{\cite{R7}} &
			\begin{tabular}[c]{@{}l@{}} A multi UAV-BS deployment scheme \\for temporary overload areas\end{tabular} & Multiple in 2D space & No & Mobile & No & No & Yes  \\ \hline
			{\cite{R21}} &
			\begin{tabular}[c]{@{}l@{}} Maximize the coverage area and the \\ served irregular dense users\end{tabular} & Multiple in 2D space & No & Mobile & No & No & Yes  \\ \hline
			{\cite{patra2019distributed}} &
			\begin{tabular}[c]{@{}l@{}} Propose a channel reassignment scheme \\ to minimize interference with irregular \\ user densities\end{tabular} & Multiple in 2D space & No & Mobile & No & No & Yes  \\ \hline  
			%
			%
			{\cite{li2022network}} &
			\begin{tabular}[c]{@{}l@{}}Maximize the energy efficiency of \\UAV-assisted cellular network\end{tabular} & Multiple in 3D space & Yes & Static & Yes & No & Yes 
			\\ \hline
			{\cite{R20}} &
			\begin{tabular}[c]{@{}l@{}}Propose a semi-progressive offloading \\ by adaptive UAV-BS parameters\end{tabular} & Single in 3D space & No & Static & No & No & Yes
			\\ \hline    
			{\cite{R22}} &
			\begin{tabular}[c]{@{}l@{}}Maximize the served users in UAV-BS \\ network when the GBS is damaged or \\ overloaded\end{tabular} & Multiple in 3D space & Yes & Static & No & Yes & Yes
			\\ \hline
			{\cite{R24}} &
			\begin{tabular}[c]{@{}l@{}}Traffic offloading and network recovery \\by a swarm of UAV-BSs\end{tabular} & Multiple in 3D space & Yes & Static & No & No & Yes
			\\ \hline
			{\cite{9454276}} &
			\begin{tabular}[c]{@{}l@{}} On-demand distributed UAV-BS \\ deployment for downlink communication\end{tabular} & Multiple in 3D space & No & Mobile & No & Yes & No
			\\ \hline
			{\cite{akarsu2018fairness}} &
			\begin{tabular}[c]{@{}l@{}}A fairness-aware multi-UAV-BS \\deployment\end{tabular} & Multiple in 3D space & No & Static & No & Yes & No  \\ \hline  
			\hspace{-1pt}{\cite{liu2018energy}} & 
			\begin{tabular}[c]{@{}l@{}} Find a control policy to maximize the \\point-of-interest (PoI) coverage score, \\fairness and minimize energy \\consumption\end{tabular} & Multiple in 2D space & No & -- & No & Yes & No  \\ \hline 
			\textbf{This paper} & \noindent\textbf{\begin{tabular}[c]{@{}l@{}} Propose an adaptive and fair \\ deployment approach to balance \\ the traffic in multi-UAV-BS networks\end{tabular}} & \textbf{Multiple in 3D space} & \textbf{Yes} & \textbf{Mobile} &
			\textbf{Yes} & \textbf{Yes} & \textbf{Yes} 
			\\ \hline
		\end{tabular}%
	}
\end{table*}

\section{Related work}
\label{sec:Related work}
This section presents essential research work related to traffic offloading and adaptive deployment of single or multiple UAV-BSs, with different goals and requirements. Table~\ref{Comparisons of Related Works and the Proposed Method} summarizes related work on UAV-BS deployment and traffic offloading issues.

%
Several studies~\cite{R7,R21,patra2019distributed} have proposed deployment methods to improve the network performance in the target area by determining the horizontal 2D position of each UAV-BS at a fixed height.
Patra~\emph{et al.}~\cite{R7} proposed a multi-UAV-BS network to provide on-demand coverage when the network is overloaded. The proposed approach follows two-fold: redistribution and then swapping UAV-BS with overloaded UAV-BS when a hotspot forms. 
Patra and Sengupta~\cite{R21} proposed a two-step multi-UAV-BS deployment mechanism to offload the traffic from temporary overloaded UAV-BS. The mechanism first deploys UAV-BSs and then applies a UAV’s dynamic positions rearrangement algorithm to reconfigure the arrangement of the UAVs for the users in the hotspot area. 
A channel reassignment scheme~\cite{patra2019distributed} was proposed to minimize interference. This scheme comprises three parts: 1) deployment of UAV-BS, 2) air to ground channel allocation to UAV-BS, and 3) reallocation of channels among UAV-BS. 

%
In addition to the above-related works for 2D UAV-BS deployment, some existing works~\cite{li2022network,R20,R22,R24,9454276} focused on 3D UAV-BS deployment to address the traffic offloading problem.
Li~\emph{et al.}~\cite{li2022network} proposed a multi-UAV-assisted transmission network, where UAV-BSs and GBS jointly transmit data using the software-defined network. This work maximizes system energy efficiency by optimizing the UAV-BS user association, UAV-BS location, and load distribution. 
Liu~\emph{et al.}~\cite{R20} proposed an adaptive UAV antenna or altitude-based model and deployed a UAV at the edge of GBS to offload the traffic. This work aims to minimize interference and the number of drones used. They also claimed that the proposed method could achieve better results without GBS; however, they did not provide any evidence. 
Omran~\emph{et al.}~\cite{R22} proposed a 3D deployment algorithm for on-demand user offloading from the malfunctioned or overloaded GBS to improve the operator’s profit with the limited user capacity of each UAV-BS. 
A greedy-based deployment approach~\cite{R24} was proposed to deploy multiple UAV-BSs in the 3D space to offload or recover the cellular network. However, they only consider static and uniformly distributed users in the target area.  
A distributed UAV-BS deployment approach for downlink communication was also proposed in~\cite{9454276} to maximize the QoS of the ground users by adaptive adjusting the altitude based on local information. 

%
Some works~\cite{akarsu2018fairness,huang2022deployment,shakoor2020joint,liu2018energy} also provide solutions to improve the performance of UAV-assisted cellular systems from the perspective of adaptive deployment and fairness issues.
A fairness-aware 3D multi-UAV-BS deployment scheme was proposed in~\cite{akarsu2018fairness} to maximize user fairness, using particle swarm optimization to achieve the best fairness performance. 
Additionally, an adaptive UAV-BSs deployment algorithm was proposed in~\cite{huang2022deployment} to provide optimal coverage for a set of ground users. The work mainly focuses on maintaining connectivity with minimizing the UAV-BS and user distance. 
In~\cite{shakoor2020joint}, a joint 3D UAV-BS deployment and path loss maximized the user coverage area. To increase the user connection time, reduce the uplink transmission power by the optimal UAV-BS deployment was proposed. 
Furthermore, a deep reinforcement learning-based UAV-BS deployment algorithm was proposed in~\cite{liu2018energy} to determine an efficient control policy to maximize coverage, fairness index value, and energy consumption.

According to the comparative summary in Table~\ref{Comparisons of Related Works and the Proposed Method}, the existing works are shown as single UAV-BS or more than one UAV-BSs deployment to offload the traffic to maximize the serving users and system capacity. Some only considered static and uniformly distributed users in the communication environment. Many of them deployed UAV-BS at a fixed altitude to cover the target area. Most existing methods only control the horizontal location (2D) or altitude (3D) of the UAV-BS, not the transmit power. The reason is that changing the horizontal position or height of the UAV-BS increases the probability of line-of-sight link with the UE. In addition, increasing the transmit power of UAV-BS may cause severe interference problems, thereby degrading system performance. Regarding the above literature, none of them jointly considered the 3D deployment space, transmission power control, and hovering energy cost of multiple UAV-BSs with the fairness constraint. Therefore, we propose a general solution that progressively follows various components to alleviate the problem of frequent overloading to satisfy the given fairness constraints.


\section{System Model}
\label{sec:System Model}
\subsection{Initial deployment and Assumptions} 
\label{Initial deployment and Assumptions}
We use the Delaunay triangulation technique (DTT)~\cite{R7} for all UAV-BSs coverage shown in Fig.~\ref{fig:fig1.pdf}, offers efficient coverage by maximizing the coverage area and minimizing the overlapping area among the UAV-BS cells.  DTT ensures no gap (coverage hole) between a group of serving UAV-BSs. Consider all users within the coverage region at the ground that follows the random distribution. Each user only uses the resource of one UAV-BS at a certain time. We consider a centralized virtual control center (CVCC) behind the GBS to help decide the association between UAV-BSs. All UAV-BSs are installed with Omni-directional antennas to transmit and receive the 4G signals in the selected environment. Our proposed model assumes dissimilar channels used by UAV-BSs; thus, interference is not considered~\cite{R23}. The GBS equips with an mm-Wave directional antenna using different devoted spectra to provide a surplus network capacity for the backhaul to all UAV-BS~\cite{6834753}~\cite{R26}. In our proposed approach, we do not consider the constraint on the backhaul.  


\subsection{Required hovering power for UAV-BS}
The UAV hovering power depends on internal and external factors. The internal factors depend on the weight of the UAV, motors, circuitry, batteries, and the weight of the payload (communication equipment). In contrast, the air density and environmental resistance are examples of external ones. The hovering power consumption of the UAV-BS is a function of the operational altitude $h_j$, defined as~\cite{R11}
\begin{align}
	p_{j}^{{\rm Hov}}=p_{0}(1+\delta)e^{{{\rm \varepsilon}h_{j}\mathord{\left/{\vphantom{{\rm \varepsilon}h_{j}2}}\right.\kern- \nulldelimiterspace}2}},   
	\label{Eq:01}
\end{align}
where $p_{0}={{\rm W}^{{3\mathord{\left/{\vphantom{3 2}}\right.\kern-\nulldelimiterspace}2}}\mathord{\left/{\vphantom {W^{{3\mathord{\left/{\vphantom{3 2}}\right. \kern-\nulldelimiterspace}2}}\sqrt{2\rho_{0}R_{{\rm nu}}D_{{\rm A}}}}}\right.\kern-\nulldelimiterspace} \sqrt{2\rho_{0}R_{{\rm nu}}D_{{\rm A}}}}$ is the power consumed by the serving UAV-BS during hover; $\delta ={D_{{\rm coff}}^{{\rm blade}}L_{{\rm ch}}\mathord{\left/{\vphantom{D_{{\rm coff}}^{{\rm blade}}L_{{\rm ch}}8S^{3}\pi R_{{\rm p}}}}\right.\kern-\nulldelimiterspace}8S^{3}\pi R_{{\rm p}}}$ is a constant; $\epsilon$ is a constant; $h_j$ is the altitude the UAV-BS; $W=W_{\rm v}$+$W_{\rm B}$+$W_{\rm P}$ is the total weight of the UAV-BS in kg; $\rho_0$ is the air density at the ocean level; $R_{\rm nu}$ is the number of rotors; $D_{\rm A}$ is a disk area; $D^\text{blade}_{\rm coff}$ is a drag coefficient; $L_{\rm ch}$ is the chord length of the UAV-BS rotor blade; $S$ is propeller advanced ratio~\cite{R12}; $R_{\rm p}$ is the radius of the propeller. Note that like batteries and motor drivers on the UAV only supply a finite amount of hovering power, ${p^\text{Hov}_{j} }$ has a physical restriction that appears as a constraint ${p^\text{Hov}_{j}}\le{p^\text{Hov}_{\rm max}}$. Table~\ref{Technical/Physical properties} presents the meanings and values of UAV-BS symbols (physical properties).

Based on the above descriptions, it is clear that the UAV-BSs’ altitude is an essential parameter for the hovering of UAV-BS in the power consumption. Equation~\eqref{Eq:01} shows that the hovering power has an exponent about the UAV-BS altitude. Thus, the UAV-BS hovering altitude can be derived from~\eqref{Eq:01} is
\begin{align}
	{h_{{\rm{min}}}}\le {h_j} = \frac{2}{\epsilon}{\rm{ln}}\frac{{p_j^{{\rm{Hov}}}} }{{{p_0}(1+\delta)}}\le{h_{{\rm{max}}}},
	\label{Eq:02}
\end{align}
where $h_{\rm{min}}$ depends on city building altitude to avoid the obstacle for collision, and $h_{\rm{max}}$ is the maximum allowable altitude to guarantee better link quality.

\begin{table}[!t]
	\centering
	\caption{Technical and Physical properties of UAV}
	\label{Technical/Physical properties}
	\begin{tabular}{|l|c|l|}
		\hline 
		\textbf{Technical/Physical Properties} & \textbf{Symbol}&\textbf{Value}
		\\\hline 
		Epsilon & $\epsilon$ & $9.7*10^{-5}$ \\ \hline 
		Vehicle weight& $W_{\rm v}$ & 10\;kg \\ \hline 
		Battery weight& $W_{\rm B}$ & 2\;kg \\ \hline 
		Payload weight& $W_{\rm P}$ & $8$\;kg \\ \hline 
		Average density of the air& $p_0$ & $1.225$\;${\rm kg/m^3}$\\\hline 
		Chord length & $L_{\rm ch}$ & $167.6*10^{-3}$\\\hline 
		Drag coefficient of the blade & $D^\text{blade}_{\rm coff}$ & $1.57*10^{-3}$ \\ \hline 
		Propeller advanced ratio~\cite{R12} & $S$ & 0.4 \\ \hline 
		Radius of propeller & $R_{\rm p}$ & $(558.2*10^{-3})/2\text{\rm m}$ \\ \hline 
		Number of rotors & $R_{\rm nu}$ & $4$ \\\hline 
		Number of batteries & $N_{\rm B}$ & $4$ \\\hline
	\end{tabular}
\end{table}

\subsection{Channel Model}  
\label{sec:system}
We consider a set of ground users, $E=\{u_1,u_2,\dots,u_N\}$, which are non-uniformly distributed in the target/hotspot area defined by $(T_{\rm A})$, as shown in Fig.~\ref{fig:fig1.pdf}. We denote $u_i=(x_i,y_i)$ as the 2 dimensional (2D) coordinates of ground users, where $i=1,2,\dots,N$, and $N$ indicates the total number of users in the system. The UAV-BS, denoted by $U_j$, is allowed to fly within predestined allowable altitudes, $h_{j}\in[h_{\min},h_{\max}]$~\eqref{Eq:02}, based on the SWaP constraints, where $j\in \{1,...,K\}$. Note that $K$ defines the maximum number of UAV-BSs in the considered $T_{\rm A}$. The 3D location of a UAV-BS $U_{j}=(x_{j},y_{j},h_{j})$, where $j\in \{1,...,K\}$. Thus, the horizontal distance between UAV-BS $U_{j}$ and ground user $u_{i}$ location, can be define as
\begin{align}
	r_{i,j}=\sqrt{(x_{j}-x_{i})^2+ (y_{j}-y_{i})^2}.
	\label{Eq:03}
\end{align}
Based on equation~\eqref{Eq:03}, the Euclidean distance between UAV-BS $U_{j}$ and ground user $u_{i}$ can be defined as
\begin{align}
	d_{i,j}= \sqrt{r_{i,j}^2 +h_{j}^2}.
	\label{Eq:04}
\end{align}

In this work, we take the air to ground channel model from~\cite{R5}, which shows the path losses of the line of sight (LoS) and non-line of sight (NLoS) are
\begin{align}
	PL^\text{LoS}_{h_{j},r_{i,j}}&=20\log_{10}\left(\frac{4\pi f_c d_{i,j}}{c}\right)+ \eta_{{\rm{LoS}}},	\nonumber\\
	PL^\text{NLoS}_{h_{j},r_{i,j}}&=20\log_{10}\left(\frac{4\pi f_c d_{i,j}}{c}\right)+ \eta_{{\rm{NLoS}}},  \nonumber
\end{align}
where $\eta_{{\rm{LoS}}}$ and $\eta_{{\rm{NLoS}}}$ are the additional mean losses~\cite{7037248} due to LoS and NLoS communication links, respectively; 
$c$ is the speed of light; $f_c$ is the carrier frequency. Therefore, we can obtained the probability of LoS signals from UAV-BS $U_{j}$ to ground user $u_{i}$ by
\begin{align}
	P^\text{LoS}_{h_{j},r_{i,j}}= &\frac{1}{1+a\exp\left(-b\left (\frac{180}{{\rm \pi }}\theta_{i,j}-a\right)\right)},  \nonumber
\end{align}
where $\theta_{i,j}=\tan^{-1}\left(\frac{h_{j}}{r_{i,j}}\right)$ (radians) is the elevation angle of the UAV-BS; $a$ and $b$ are the constant factors depending on the different environmental conditions (rural, urban, dense urban, etc.)~\cite{R5}. 
With $P^\text{LoS}_ {h_{j}, r_{i,j}}$, the probability of NLoS signals from UAV-BS $U_{j}$ to ground user $u_{i}$ is $P^\text{NLoS}_{h_{j},r_{i,j}} =1-P^\text{LoS}_{h_{j},r_{i,j}}$. In summary, the average path loss of the signal from UAV-BS $U_{j}$ to ground users $u_{i}$ will be
\begin{align}
	PL_{{h_j},{r_{i,j}}}^{{\rm{Avg}}}&=P_{{h_j},{r_{i,j}}}^{{\rm{LoS}}}\times PL_{{h_j},{r_{i,j}}}^{{\rm {LoS}}}+P_{{h_j},{r_{i,j}}}^{ {\rm{NLoS}}}\times PL_{{h_j},{r_{i,j}}}^{{\rm{NLoS}}}\notag \\
	&=\frac{A}{{1+a\exp\left({-b\left[{{\frac{180}{{\rm \pi }}\theta_{i,j}}-a}\right]}\right)}}+20{\log_{10}}\left ({{d_{i,j}}}\right)+\beta, 
	\label{Eq:05} 
\end{align}
where $\beta=20\log_{10}\left(\frac{4\pi f_c}{c}\right) +\eta_{{\rm{NLoS}}}$ and ${A = \eta_{{\rm{LoS}}}}-{\eta_{{\rm{NLoS}}}}$.

Let $p_{i,j}$ be the minimum required transmit power for transmitting signal from the $j$-th UAV-BS $U_{j}$ to ground user $u_{i}$, where $i\in\{1,2,\dots,N\}$ (see appendix-\ref{appendix-A}). For the successful signal transmission, the received signal-to-noise ratio (SNR) $\gamma _{i,j}$, at a user should be larger than the predefined SNR threshold, $\gamma_{\rm th}$. Thus, the SNR for the user $u_i$ associated with the $j$-th UAV-BS $U_{j}$ can be define as
\begin{align}
	\gamma_{i,j} =\frac{p_{i,j} .10^{-{PL_{h_{j} ,r_{i,j} }^{{\rm Avg}} \mathord{\left/ {\vphantom {PL_{h_{j} ,r_{i,j} }^{{\rm Avg}}  10}} \right. \kern-\nulldelimiterspace} 10} } }{B_{i,j} \sigma ^{2} } \ge \gamma_{\rm th},
	\label{Eq:06} 
\end{align}
where $j\in\{1,2,...,K\}$.
To represent whether ground user $u_i$ is associated with the $j$-th UAV-BS $U_{j}$ or not, let $\zeta_{i,j}$ be the indicator function as follows:
\begin{equation}
	\zeta_{i,j} =\left\{\begin{array}{l} {1,\, \, \, \, \, \, \, \, {\rm if}\, \, \gamma_{i,j} \, \geq \gamma_{\rm th} \wedge \gamma_{i,j'}<\gamma_{i,j}},\forall j'\neq j; \\ 
		{0,\, \, \, \, \, \, \, {\rm otherwise}.} \end{array}\right.  
	\label{Eq:07}
\end{equation} 
Default, each user $u_i$ is associated with the $j$-th UAV-BS $U_{j}$ to achieve the best SNR value, $\gamma _{i,j}$. We also assume that each user $u_i$ can only connect to one UAV-BS $U_{j}$ at a time and such a constraint can be written as
\begin{align}
	\sum_{j=1}^{K}\zeta_{i,j}=1, 
	\label{Eq:08} 
\end{align}
where $i\in\{1,2,\dots,N\}$.

The allocated data rate (in mbps) of user associated with UAV-BS $U_{j}$ is obtained from the Shannon theorem, expressed as
\begin{align}
	c_{i,j}=B_{i,j}\log_2(1+\gamma_{i,j})\zeta _{i,j},
	\label{Eq:09} 
\end{align}
where $B_{i,j}$ is the allocated bandwidth (in MHz) of down-link connection from the $j$-th UAV-BS $U_{j}$ to ground user $u_i$ and $j\in\{1,2,...,K\}$. The total power (communication and hover) consumption of the $j$-th UAV-BS is
\begin{align}
	p_{j}^{{\rm Total}}=\sum_{u_i\in \Omega_{j},\forall i\in\{1,2,\dots,N\}}p_{i,j}+p_{j}^{{\rm Hov}},
	\label{Eq:10} 
\end{align}
where $\Omega_{j}$ is the set of users associated with the $j$-th UAV-BS and $j\in\{1,2,...,K\}$. 
According to~\eqref{Eq:09}, the data transmission rate of UAV-BS $U_{j}$ for serving their associated users can be defined by
\begin{align}
	{C_j}=\sum_{u_i\in{\Omega_j},\forall i\in\{1,2,\dots,N\}}c_{i,j},
	\label{Eq:11} 
\end{align}
where $j\in\{1,2,...,K\}$.

With~\eqref{Eq:10} and~\eqref{Eq:11}, the energy efficiency of~$j$-${\rm th}$ UAV-BS $U_{j}$ (in bit/Joule) can be derived by 
\begin{align}
	E_{j}=\dfrac{\displaystyle\sum_{\forall u_i\in{\Omega_j},i\in\{1,2,\dots,N\}}c_{i,j}}{\displaystyle\sum_{\forall u_i\in{\Omega_j},i\in\{1,2,\dots,N\}}p_{i,j}+p_j^{\rm Hov}}, 
	\label{Eq:12} 
\end{align}
where $p_{i,j}$ can be obtained by solving the nonlinear partial differential equation, $\frac{\partial E_{j}}{\partial p_{i,j}}=0$, and $j\in\{1,\dots,K\}$. This nonlinear partial differential equation is equivalent to
\begin{align}
	\frac{PL_{{h_j},{r_{i,j}}}^{{\rm{Avg}}}}{\left(1+\frac{p_{i,j}\cdot PL_{{h_j},{r_{i,j}}}^{ {\rm{Avg}}}}{B_{i,j}\sigma^2}\right)\left(p_{i,j} +P_{j}^\text{Hov}\right)\sigma^2(\ln2)}\notag \\-\frac{B_{i,j}\log_{2}\left(1+\frac{p_{i,j}\cdot PL_{{h_j},{r_{i,j}}}^{{\rm{Avg}}}} {B_{i,j} 
			\sigma ^2}\right)}{\left(p_{i,j}+P_{j}^\text{Hov}\right)^{2}}=0.  
	\label{Eq:13} 
\end{align}
Appendix~\ref{appendix-B} provides detailed proof of~\eqref{Eq:13}.

The fairness among the users can be shown by the fairness metric named Jain’s fairness index (JFI) denoted by {$\xi$} was proposed by R.K. Jain~\cite{R19}, given as follows:
\begin{align}
	\xi =\frac{\left(\sum_{j=1}^{K}C_{j}\right)^{2} }{K\sum _{j=1}^{K}C_{j}^{2}}.
	\label{Eq:14} 
\end{align} 
The fairness index should be limited, which can be a proportion between 0 and 1. The higher value of the fairness index is the smaller differences between the allocated total data rates and users. In this work, we also consider JFI as an important constraint in the formulated optimization problem.


\section{Problem formulation} 
\label{sec: Problem Formulation}
In this work, we consider the deployment of multi-UAV-BSs in the target/desired location to improve the energy efficiency of the UAV-BS network. The deployment of UAV-BS must satisfy the predefined minimum data rate requirements~\eqref{Eq:09}. Due to the on-board battery capacity of UAVs, hovering, communication equipment, and round-trip recharging waste time and energy. Long endurance and reliable communication are desirable in critical scenarios, such as disaster locations and extended communications. Long and reliable UAV-assisted communication needs to improve energy efficiency. Therefore, we aim to maximize energy efficiency by optimizing the altitude and transmit power allocation of the UAV-BS. We refer to such a problem as \textit{maximizing the total energy efficiency of multi UAV-BSs} (MTEU) problem, which can be defined as follows.
\begin{definition}[MTEU problem]\label{def:p1}
	With the above-defined notations and assumptions, the MTEU problem is to use the given number of available UAV-BSs to find the appropriate altitude and transmit power such that
	\begin{align}
		{\mathop{\max}\limits_{\begin{array}{c} {h_{j} ,p_{i,j} } \end{array}}}\sum _{j=1}^{K}E_{j}=&{\mathop{\max}\limits_{\begin{array}{c}{h_{j},p_{i,j}}\end{array}}}\sum_{j=1}^{K}\frac{\displaystyle\sum_{\forall u_i\in{\Omega_j},i\in\{1,2,\dots,N\}}c_{i,j}  }{\displaystyle\sum_{\forall u_i\in{\Omega_j},i\in\{1,2,\dots,N\}}p_{i,j}+P_{j}^{{\rm Hov}}} 
		\tag{P1}
		\label{Eq:obj_P1}\\
		\hspace{3em}\text{subject to}&~\eqref{Eq:06},~\eqref{Eq:07},~\eqref{Eq:08},~\eqref{Eq:09},\nonumber\\
		&~{h_{\min}}\le{h_j}\le{h_{\max}},
		\label{Eq:15} \\
		&~{p_{\min}}\le {p_{i,j}}\le{p_{\max}}, 
		\label{Eq:16} \\
		&~0\le {|\Omega_{j}|}\le{\omega_{\max}},
		\label{Eq:17}\\
		&~\xi\ge\xi_{\rm th}.
		\label{Eq:18} 
	\end{align}
\end{definition}

In constraint~\eqref{Eq:15}, the multi-UAV-BS deployment always flies within allowable limits~[$h_{\rm min}$,\;$h_{\rm max}$]. 
Allowed altitudes are usually determined by local laws (usually higher than city building heights) and the hovering capability of the UAV. 
The constraint~\eqref{Eq:16} shows the limitation of the transmission power of each UAV-BS for serving the associated users in the coverage. The transmission power usually depends on UAV altitude, service time, and user association number. The constraint~\eqref{Eq:17} shows the user association limitation of the UAV-BS. The constraint~\eqref{Eq:18} shows the guaranteed fair user distribution limitation at each UAV-BS. 

In the considered system, the allowable path-loss of each user is a predefined and fixed value, $PL^{\rm{allow}}$. Using $PL^{\rm{allow}}$ and~\eqref{Eq:05}, we can compute the optimal angle, $\theta _{j}^{{\rm opt}}$, by solving the nonlinear partial differential equation, $\frac{\partial r_{i,j}}{\partial \theta _{j}^{{\rm opt}}}=0$, of~\eqref{Eq:05}, which can be expressed as~\cite{R5}
\begin{align} 
	{\frac{{\rm \pi }\tan \theta_{j} }{9\ln \left(10\right)} +\frac{abA\exp \left(-b\left[\frac{180}{{\rm \pi }} \theta_{j} -a\right]\right)}{\left(a\exp \left(-b\left[\frac{180}{{\rm \pi }} \theta_{j} -a\right]\right)+1\right)^{2} } =0}.
	\label{Eq:19} 
\end{align} 
With the obtained optimal angle $\theta _{j}^{{\rm opt}}$, if altitude of the $j$-th UAV-BS, $h_j$, is given, the corresponding coverage, $R_j$, can be derived by
\begin{align} 
	\theta_{j}=\tan^{-1}\left(h_j/R_j\right),
	\label{Eq:20} 
\end{align} 
and $\theta_{j}$ is set to optimal angle $\theta _{j}^{{\rm opt}}$.
Because the maximum altitude, $h_{\max}$, is a predefined constraint~\eqref{Eq:02}, we can use~\eqref{Eq:20} to get the maximum and allowable coverage radius, $R_{\max}$, provided by a UAV-BS. Additionally, the maximum allowable LoS distance (Euclidean distance) between the $j$-th UAV-BS will be $d_{\max}=R_{\max}\sec\theta _{j}^{{\rm opt}}$.

\subsection{Feasibility Analysis}
\label{sec: Feasibility Analysis}
The proposed MTEU problem~\eqref{Eq:obj_P1} is always feasible while $\zeta_{\rm th}$ is close to $0$ and $\omega_{\max}$ is set to a large constant. Let us discuss a general example of the problem. Determine the altitude and transmission power of each UAV-BS that do not violate constraints~\eqref{Eq:15} and~\eqref{Eq:16}, and user association of UAV-BS also does not violate~\eqref{Eq:17}. According to our system model and assumptions~\eqref{Eq:07} and~\eqref{Eq:08}, each user will always be associated with a UAV-BS. Therefore,~\eqref{Eq:obj_P1} will not be 0. However, the association capacity of each UAV-BS~\eqref{Eq:17} and the fairness constraint~\eqref{Eq:18} may be difficult to achieve because users may be unevenly distributed. In other words, the service provider may not provide a feasible deployment for satisfying the given constraints $\omega_{\max}$ and $\zeta_{\rm th}$. In this case, the constraints $\omega_{\max}$ and $\zeta_{\rm th}$ need to be relaxed to search feasible deployment parameters, $\{h_{i,j}\}$ and $\{p_{i,j}\}$.

\subsection{NP-Hardness}
\label{sec: NP-Hardness}
This section will show that the considered MTEU problem~\eqref{Eq:obj_P1} is NP-hard. To verify this, we relax some constraints (fixed altitude) and modify the MTEU problem into a transmit power optimization (TPO) problem as~\eqref{Eq:obj_P2}. The TPO problem is a special case of the MTEU problem, while the altitude of each UAV-BS is fixed. Next, if the TPO problem is NP-hard/NP-complete, it proves the considered MTEU problem is NP-hard. The TPO problem is defined as follows.

\begin{definition}[TPO Problem]\label{def:p2}
	With the above-defined notations and assumptions, if we ignore the constraint on altitude limitation and make all UAV-BSs only fly at the same/fixed altitude, problem~\eqref{Eq:obj_P1} will be simplified as the following TPO problem
	\begin{align}
		{\mathop{\max }\limits_{\begin{array}{c} {p_{i,j} } \end{array}}}& \sum _{j=1}^{K}\frac{\displaystyle\sum_{\forall u_i\in{\Omega_j},i\in\{1,2,\dots,N\}}c_{i,j} }{\displaystyle\sum_{\forall u_i\in{\Omega_j},i\in\{1,2,\dots,N\}}p_{i,j} +P_{j}^{{\rm Hov}} }   
		\tag{P2}
		\label{Eq:obj_P2}\\
		\text{subject to}&~\eqref{Eq:06},~\eqref{Eq:07},~\eqref{Eq:08},~\eqref{Eq:09},~\eqref{Eq:16},~\eqref{Eq:17},~\eqref{Eq:18}. \nonumber
	\end{align}
\end{definition} 

Using Definition~\ref{def:p2}, we deduce the following theorem.

\begin{thm}\label{thm:1}
	The MTEU problem is NP-hard.
\end{thm}

\begin{proof}
	With the relaxation of altitude constraints and all UAV-BSs flying at a fixed altitude, the TPO problem~\eqref{Eq:obj_P2} is an NP-hard problem, as proved in~\cite{7523951}. Therefore, the TPO problem is NP-hard, implying that the MTEU problem is an NP-hard problem.
\end{proof}

\begin{figure*}[!h]
	\centering	
	\subfigure[]{
		\label{fig:fig2:a}               
		\includegraphics[width=.24\textwidth]{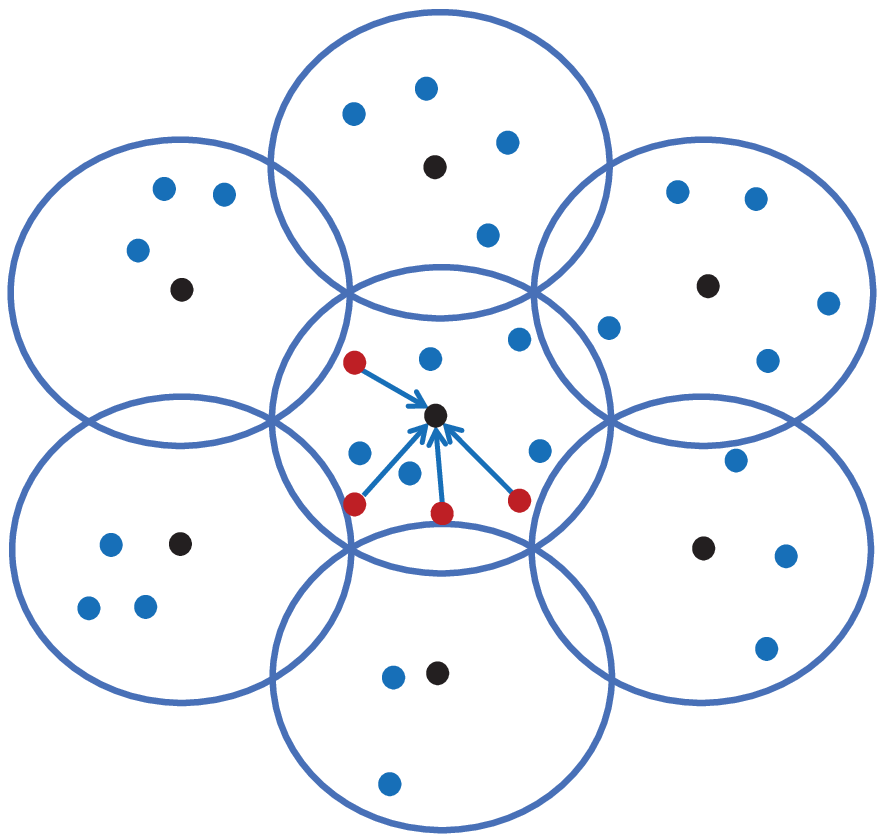}}\hspace{1em}%
	\subfigure[]{
		\label{fig:fig2:b}               
		\includegraphics[width=.24\textwidth]{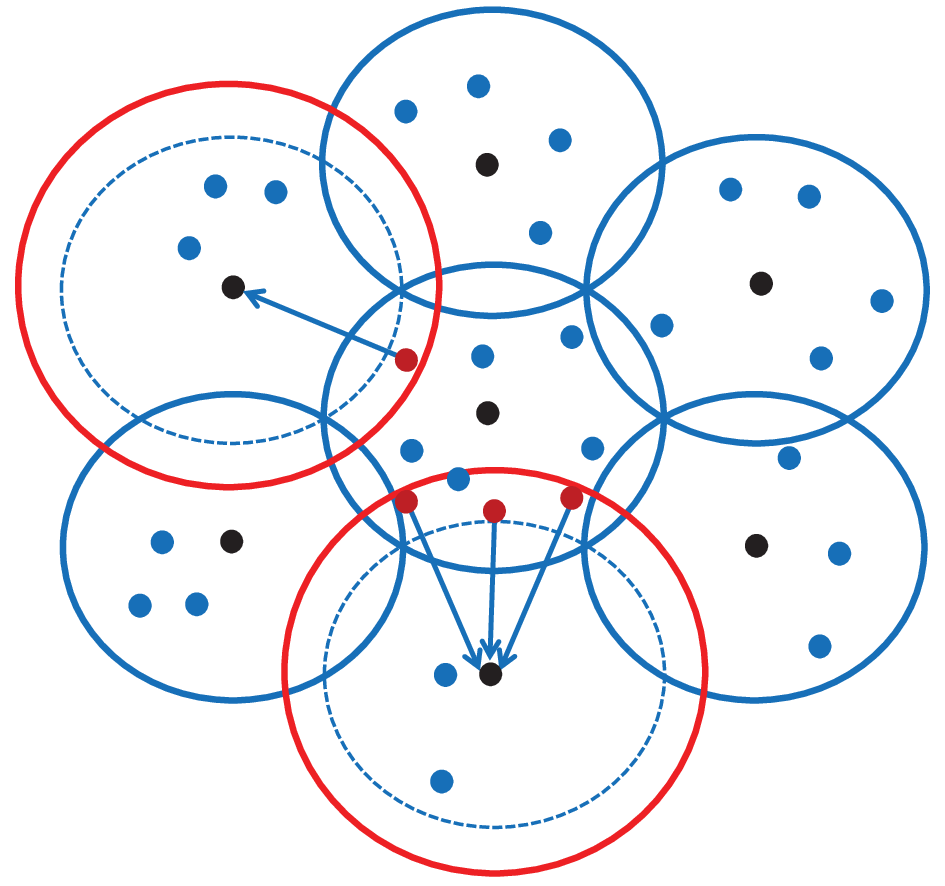}}\hspace{1em}%
	\subfigure[]{
		\label{fig:fig2:c}               
		\includegraphics[width=.42\textwidth]{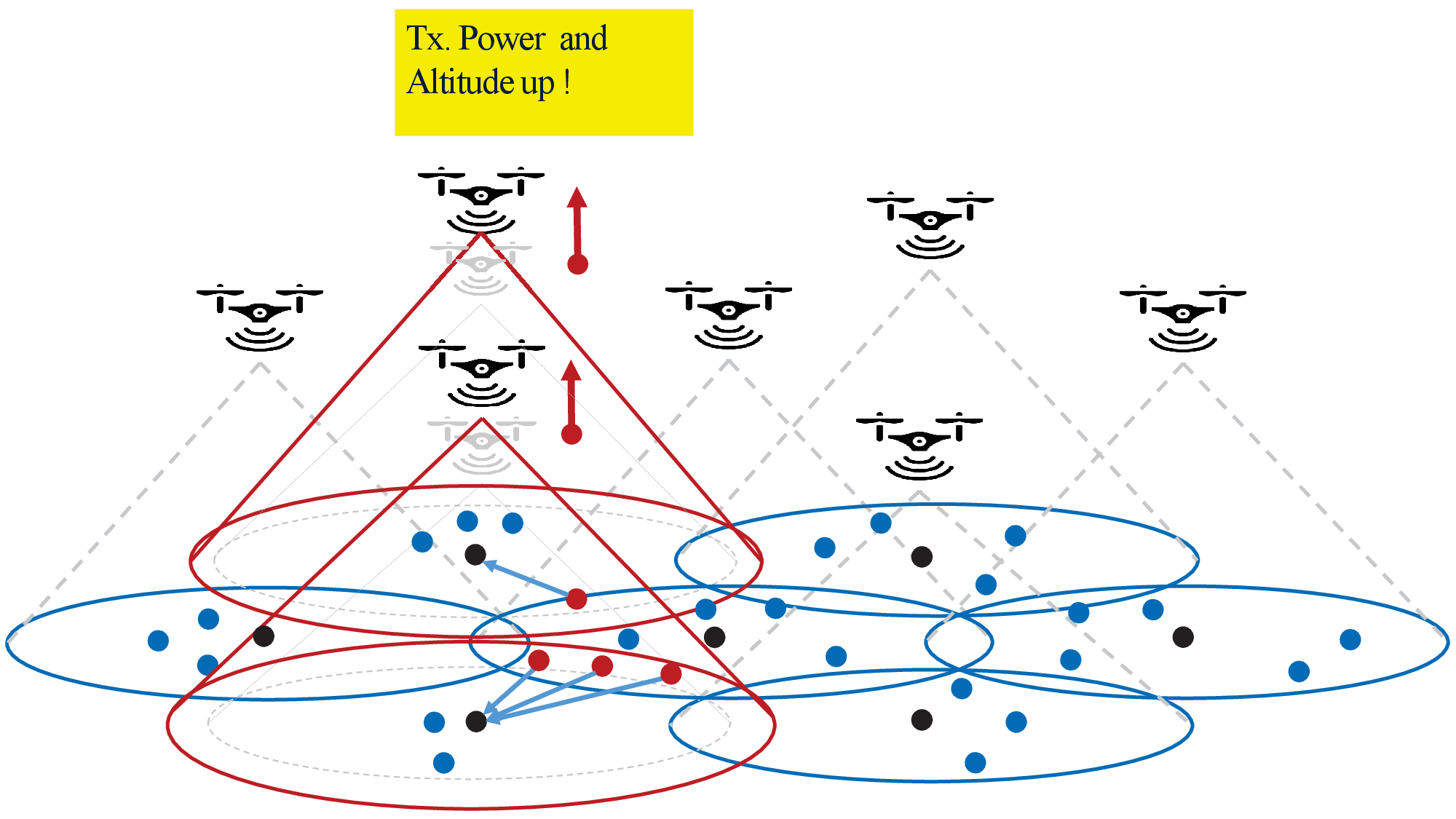}}%
	\caption{The 2D and 3D visualized views of different phases of the proposed AFD:~\subref{fig:fig2:a} Initial phase,~\subref{fig:fig2:b} re-association phase, and~\subref{fig:fig2:c} altitudes and power optimization phase. The blue dot shows the user; the red dot shows the excess users in the central cell; the black dot shows all UAV-BS center locations; the blue circle and blue dash circle show the original coverage of UAV-BS. The arrows show the association direction of the users; the red circle shows the new coverage of UAV-BS after applying AFD.}
	\label{fig:fig2}                     
\end{figure*}

\begin{figure}[!t]
	\centering
	\includegraphics[width=\columnwidth]{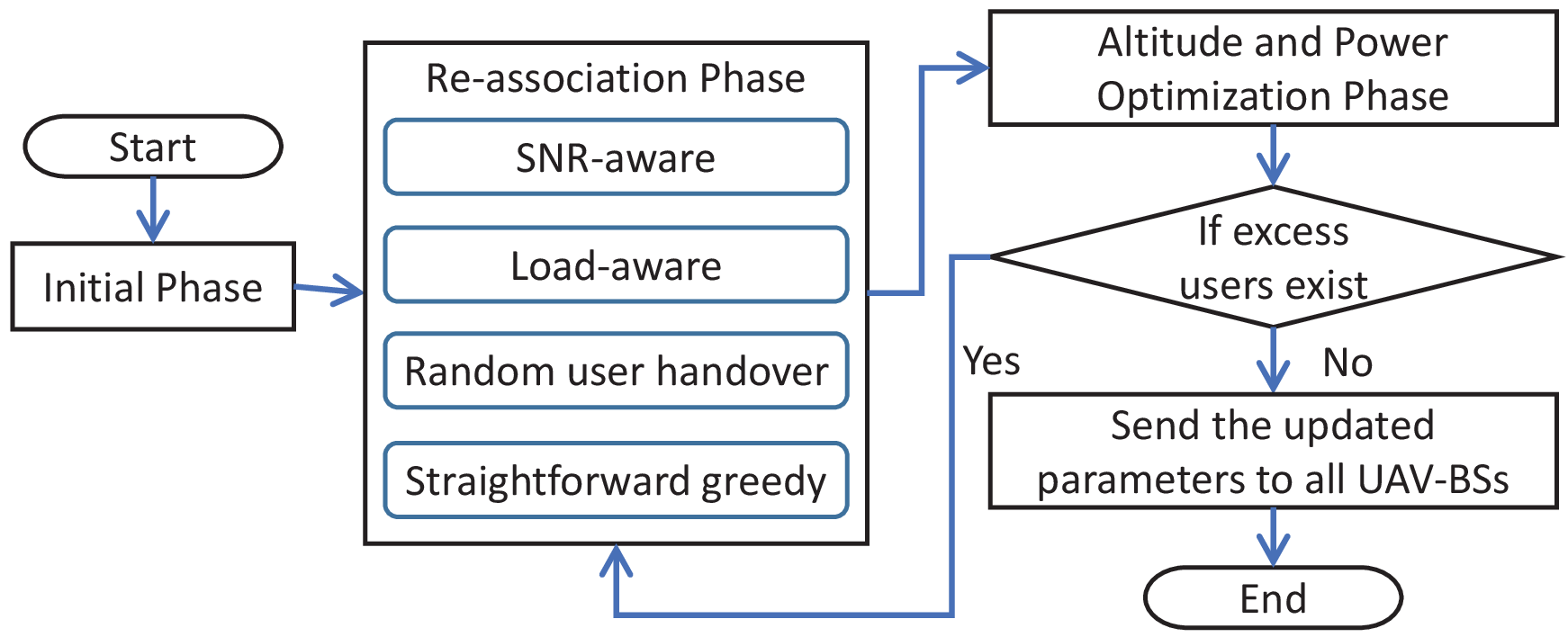}
	\caption{The flowchart of the AFD framework}
	\label{fig:fig3}
\end{figure}

	\section{The Proposed Adaptive and Fair Deployment Approach (AFD)} 
	\label{sec:AFD}
	In this section, we first introduce the main idea and overview framework of the proposed AFD method. We then describe several re-association schemes applied in the proposed AFD framework. After that, the procedure of AFD will be described in detail. Finally, we will discuss the benefits of our design.
	
	\subsection{The Main Idea and Framework of AFD}
	To solve the considered MTEU problem~\eqref{Eq:obj_P1}, the proposed AFD should meet the fair user distribution~\eqref{Eq:18} at each UAV-BS and guarantee data rates~\eqref{Eq:06} for all users in the considered target area. The objective of our approach is to provide an energy-efficient and adaptive deployment of UAV-BSs for fair traffic offloading. The proposed AFD approach has three phases and the 2D/3D visualized view of each phase is presented in Fig.~\ref{fig:fig2}. The overview of flowchart of the proposed AFD is shown in Fig.~\ref{fig:fig3}. The responsibilities of each phase are described as follows:
	\begin{enumerate}
		\item \textbf{Initial Phase:} In this phase, the CVCC will load the some predefined parameters and prepare some information in advance for the forthcoming computation. With the above information, as shown in Fig.~\ref{fig:fig2:a}, the CVCC will check the load of each UAV-BS and determines the number of excess users of the overloaded UAV-BS. CVCC will perform the initial phase only once, and then run the following two phases repeatedly until there are no excess users in multiple UAV-BS systems.
		\item \textbf{Re-association Phase:} In this phase, as presented in Fig.~\ref{fig:fig2:b}, the CVCC will compute the decision to re-associate excess users of the overloaded UAV-BS to neighboring available UAV-BSs. Note that only one user will be re-associated at a time. In the proposed AFD framework, we implement three re-association schemes. Each scenario will be described in detail in the next subsection.
		\item \textbf{Altitude and Power Optimization Phase:} Based on each re-association of the user in the previous stage, the CVCC will calculate the minimum altitude, minimum hover power, and minimum transmit power required for the updated neighboring UAV-BS. The CVCC will repeat the above phases for each excess user until no excess users exist.
		Finally, the CVCC will send the updated final parameters to all UAV-BSs to adjust their altitude and transmit power. This optimization is visualized as a 3D view in Fig.~\ref{fig:fig2:c}.
	\end{enumerate}

	\begin{figure*}[!t]
		\centering	
		\subfigure[]{
			\label{fig:implementation:1} 
			\includegraphics[width=0.24\textwidth]{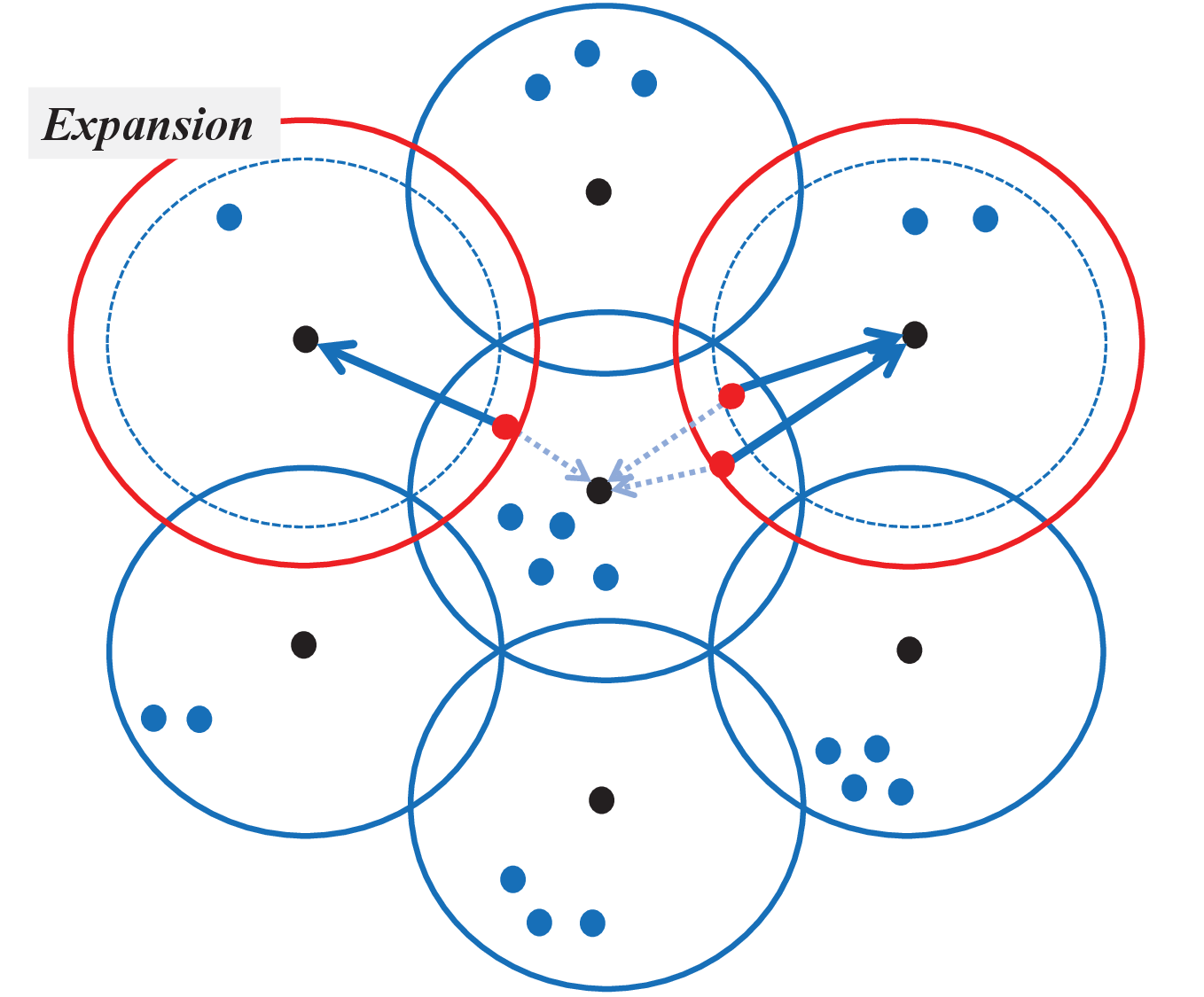}
		}%
		\subfigure[]{
			\label{fig:implementation:2} 
			\includegraphics[width=0.24\textwidth]{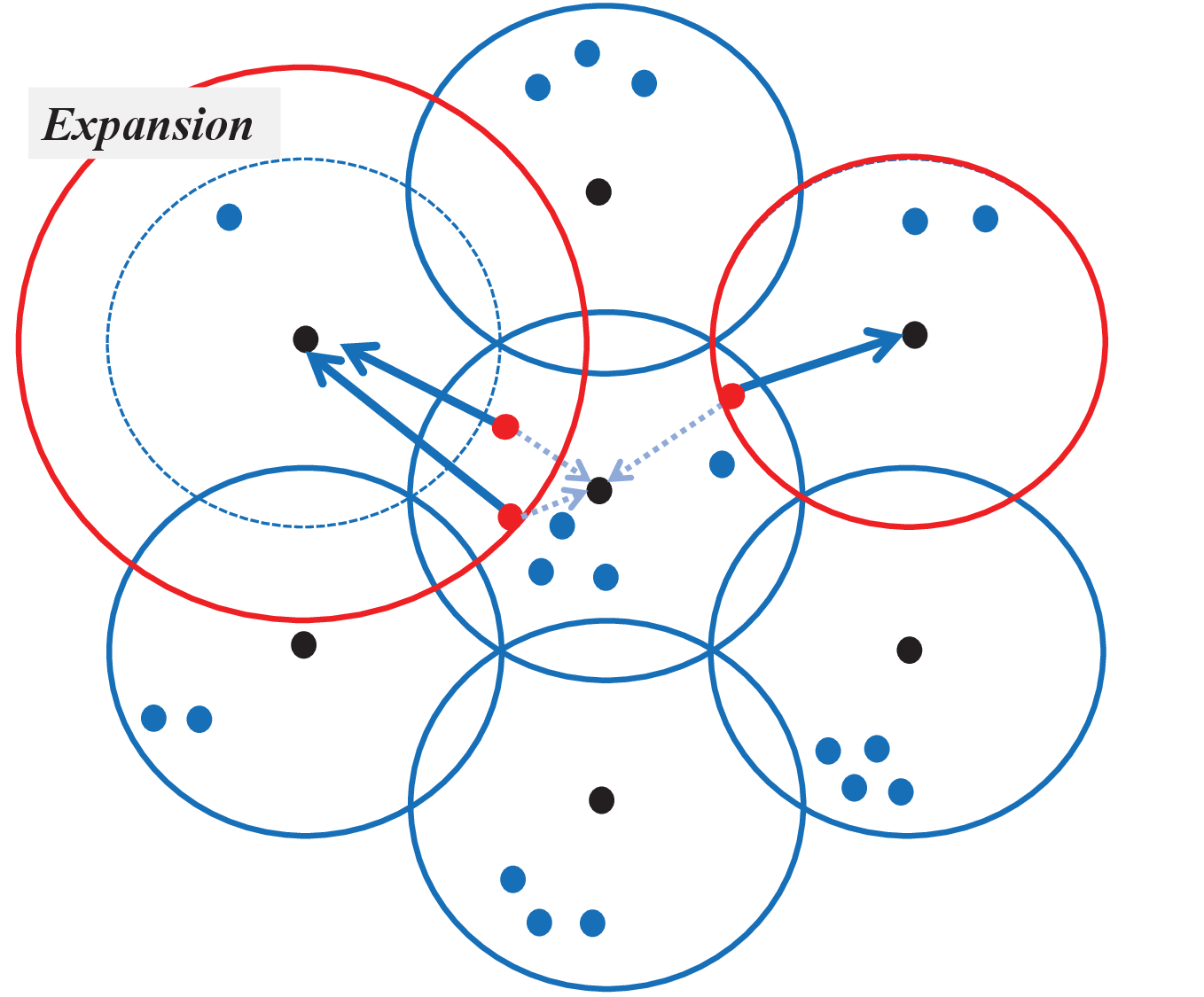}
		}%
		\subfigure[]{
			\label{fig:implementation:3} 
			\includegraphics[width=0.24\textwidth]{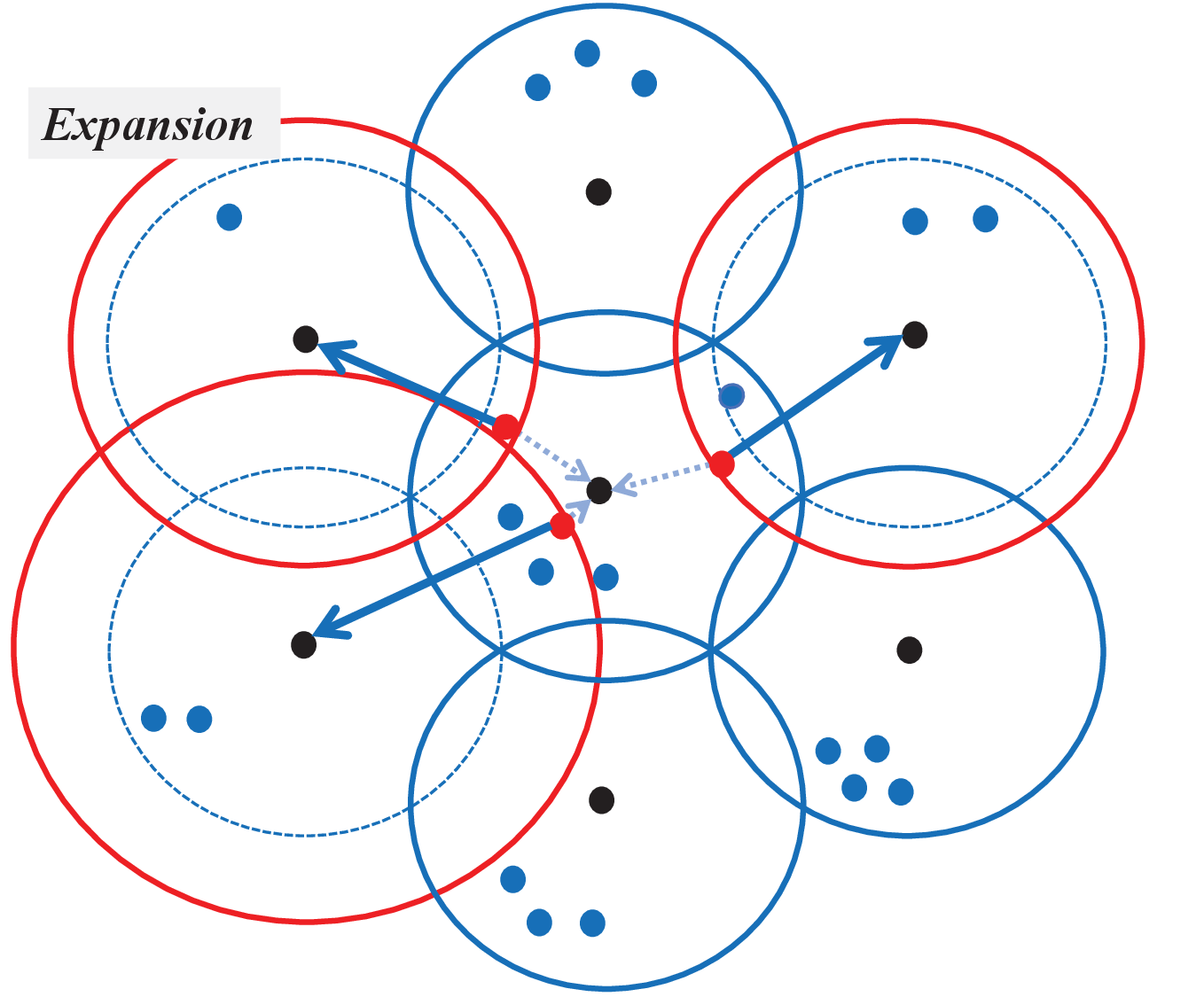}
		}%
		\subfigure[]{
			\label{fig:implementation:4} 
			\includegraphics[width=0.24\textwidth]{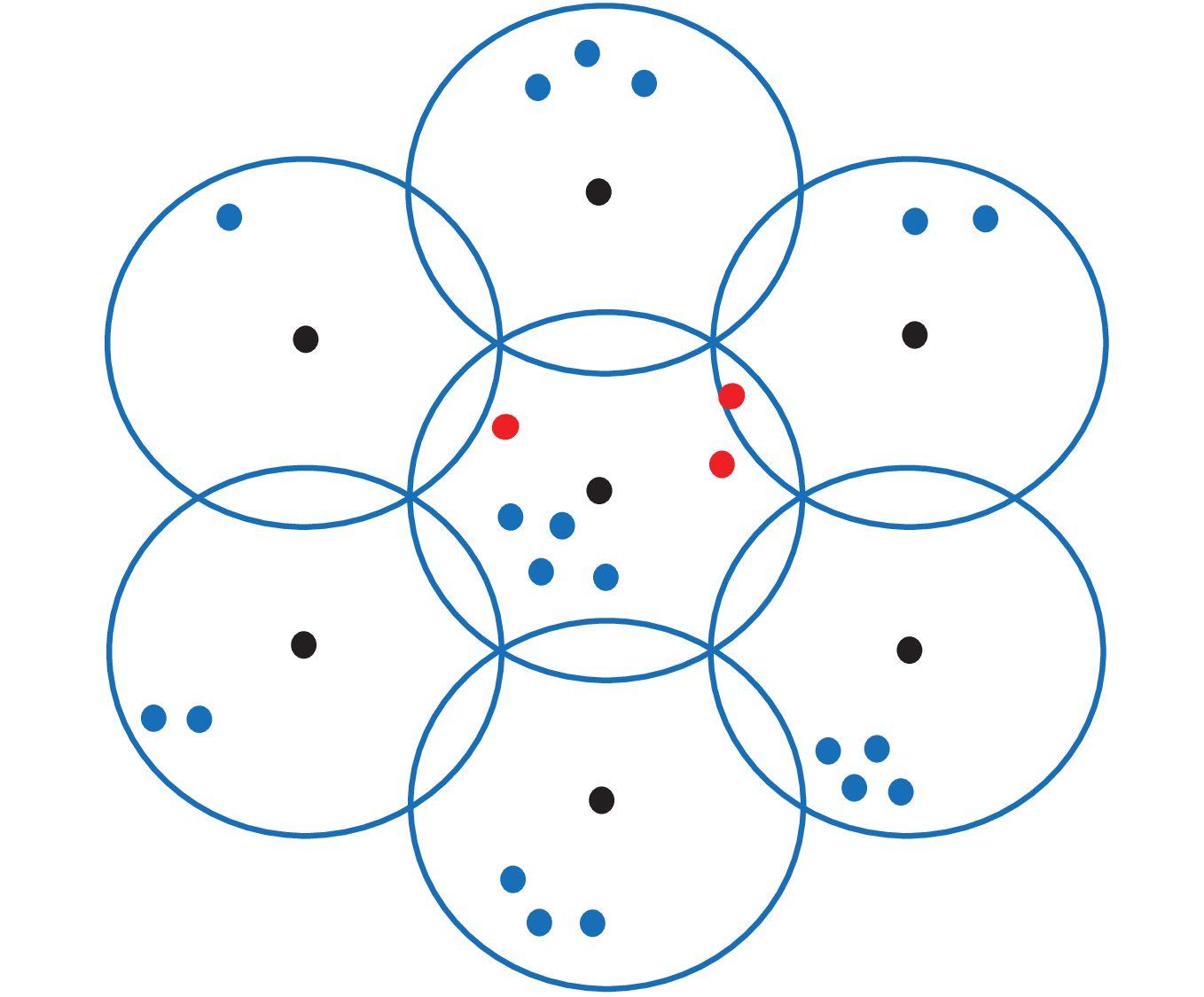}
		}%
		\caption{Re-association schemes of AFD:~\subref{fig:implementation:1} SNR-aware,~\subref{fig:implementation:2} load-aware,~\subref{fig:implementation:3} random user handover, and~\subref{fig:implementation:4} straightforward greedy. The blue dot shows the user location; the red dot shows the excess users’ location in the central cell; the black dot shows all UAV-BS center locations; the blue circle and blue dash circle show the original coverage of UAV-BS; the light dash arrow direction shows the previous user association direction; the blue arrow direction shows the next user association direction after traffic offload; the red circle shows the new coverage of UAV-BS after applying AFD.
		}
		\label{fig:implementation}   
	\end{figure*}

	\subsection{Re-association Schemes}
	The proposed AFD framework includes three proactive re-association schemes and one default straightforward scheme. As shown in the Fig.~\ref{fig:implementation}, we describe each scheme in detail as follows:
	\begin{enumerate}
		\item \textbf{SNR-aware re-association.} For the SNR-aware re-association scheme shown in Fig.~\ref{fig:implementation:1}, the CVCC first selects the user with the minimum/worst received SNR from the overloaded UAV-BS. Next, CVCC finds the UAV-BS closest to the selected user, except for the overloaded UAV-BS.
		
		\item \textbf{Load-aware re-association.} As shown in Fig.~\ref{fig:implementation:2}, comparing to SNR-aware re-association scheme, the CVCC using the load-aware re-association scheme first selects the neighboring UAV-BS that has the maximum available user associations. Then, the CVCC selects the user closest to the selected neighboring available UAV-BS from the users covered by the overloaded UAV-BS. Next, the CVCC adjusts the height and transmit power of the selected available UAV-BS to expand its coverage, thereby establishing an association between the selected user and the neighboring available UAV-BS. 
		
		\item \textbf{Random user handover re-association.} In the random user re-association scheme, as shown in Fig.~\ref{fig:implementation:3}, the CVCC first randomly selects a neighboring UAV-BS. Then, the CVCC selects the user nearest to the selected neighboring UAV-BS from the users covered by the overloaded UAV-BS. 
		
		\item \textbf{Straightforward greedy.} Compare to the above proactive scheme, the last one is a passive scheme. As shown in Fig.~\ref{fig:implementation:4}, the CVCC will not help with the user association process. Each user follows a straightforward greedy strategy to associate the nearby UAV-BS that provides the best SNR value signal.
	\end{enumerate}
	
	After finishing the re-association, the CVCC enters the final phase and uses the updated user association decision to calculate the adjusted altitude, minimum hovering power, and minimum transmit power required for all UAV-BSs. Then, all adjusted parameters of altitudes, minimum hovering power, and minimum transmit power will be sent to all UAV-BSs. All UAV-BSs will sequentially adjust altitude and transmit power based on received parameters to extend coverage and complete updated user associations. The detailed process will be introduced in the next subsection.

	\begin{algorithm*}[!t]
		\caption{The main procedure of AFD} 
		\label{alg:AFDA}
		\small
		\begin{multicols}{2}
			\textbf{Input:} \\
			\hspace*{1em} $\omega_{\max}$: the given maximum number of associations per UAV-BS; \\
			\hspace*{1em} $h_{\min}$: the given minimum altitude of UAV-BS; \\		
			\hspace*{1em} $h_{\max}$: the given maximum altitude of UAV-BS; \\
			\hspace*{1em} $r\_s\in\{0,1,2\}$: the given parameter to select SNR-aware or load-aware re-association scheme, where $1\leq \forall j \leq K$; \\
			\hspace*{1em} ${\rm{\bf E}} = \left\{ {{u_1},\dots,{u_i},\dots ,{u_N}} \right\}$: the set of horizontal locations of all the users, where $1\leq \forall i \leq N$; \\	
			\hspace*{1em} ${\rm{\bf U}} = \left\{ {{U_1},\dots,{U_j},\dots ,{U_K}} \right\}$: the set of horizontal locations of all the UAV-BSs; \\		
			\hspace*{1em} ${\rm\bf\Omega}=\{{\rm\bf\Omega}_1,\dots,{\rm\bf\Omega}_K\}$: the association map of all UAV-BSs;\\
			\hspace*{1em} ${\rm\bf\Omega_{j}}$: the set of users associated with the $j$-th UAV-BS, $\forall 1\leq j\leq K$;\\
			\hspace*{1em} $|{\rm\bf\Omega_{j}}|$: the number of users associated with the $j$-th UAV-BS.\\		
			\hspace*{1em} ${\rm{\bf H}} = \left\{ {{h_1},\dots,{h_j},\dots ,{h_K}} \right\}$: the set of altitudes of all the UAV-BSs; \\ 
			\hspace*{1em} /* $h_{j}$: the altitude of the $j$-th UAV-BS, $1\leq\forall j\leq K$. */ \\	
			\textbf{Pseudo-code:}
			\begin{algorithmic}[1]
				\STATE Let ${\rm\bf M}=\{m_1,m_2\dots,m_K\}$ be a array to mark whether each UAV-BS has changed its altitude, where $m_j\in\{0,1\},1\leq\forall j\leq K$; \label{alg:AFDA:01}	
				
				\STATE Let ${\rm\bf R}=\{r_{1,1},\dots,r_{i,j},\dots,r_{N,K}\}$ be a set to record horizontal distances from all UAV-BSs to all users by using~\eqref{Eq:03} with $\rm\bf E$ and $\rm\bf U$;  \label{alg:AFDA:02}
				
				\STATE Let ${\rm\bf D}=\{d_{1,1},\dots,d_{i,j},\dots,d_{N,K}\}$ be a matrix to record the Euclidean distances from all UAV-BSs to all users by using~\eqref{Eq:04} with $\rm\bf R$ and the initial minimum altitude $h_{\min}$; 	\label{alg:AFDA:03}
				
				\STATE Let ${\rm\bf P}=\{p_{1,1},\dots,p_{i,j},\dots,p_{N,K}\}$ be the set of transmission power from the $j$-th UAV-BS to the $i$-th user; 	\label{alg:AFDA:04}
				
				\STATE Let ${\rm\bf P^{Hov}}=\{p_1^{\rm Hov},\dots,p_K^{\rm Hov}\}$ be the required hovering power of each UAV-BS; 	\label{alg:AFDA:05}
				
				\FOR{$j=1$ to $K$} \label{alg:AFDA:06}
				
				
				\WHILE {$|{\rm \bf\Omega_{j}}|-\omega_{\max}>0$} \label{alg:AFDA:07}			
				
				\STATE Invoke Algorithm~\ref{alg:reassociation_function} with $r\_s$ to get the selected user $i^{*}$ and UAV-BS $j^{*}$;\label{alg:AFDA:08} \\
				/* Altitude optimization starts */
				
				\STATE Calculate $\theta_{j^{*}}^{\rm opt}$ (degree) by~\eqref{Eq:19}; \label{alg:AFDA:09}
				
				\STATE Calculate $h_{j^{*}} = {\rm\bf R}\left({i^{*},j^{*}}\right)\tan\theta_{j^{*}}^{\rm opt}$; \label{alg:AFDA:10}
				
				\IF{${h_{j*}} \geq {h_{\max}}$}  \label{alg:AFDA:11}
				
				\STATE ${h_{j^{*}}} = {h_{\max}}$; \label{alg:AFDA:12}
				
				\ELSIF{${h_{j^{*}}} \leq {h_{\min}}$} \label{alg:AFDA:13}
				
				\STATE ${h_{j^{*}}} = {h_{\min}}$; \label{alg:AFDA:14}
				
				\ENDIF \label{alg:AFDA:15}
				
				\STATE ${\rm\bf\Omega_j}={\rm\bf\Omega_j}\setminus u_{i^{*}}$; \label{alg:AFDA:16}	
				
				\STATE ${\rm\bf\Omega_{j^{*}}}={\rm\bf\Omega_{j^{*}}}\cup u_{i^{*}}$; \label{alg:AFDA:17}	
				
				\STATE ${\rm{{\bf H}}}\left({j^{*}}\right)={h_{j^{*}}}$; \label{alg:AFDA:18}					
				\hfill// Commit the updated ${h_{j^{*}}}$
				
				\ENDWHILE \label{alg:AFDA:19}
				
				\ENDFOR \label{alg:AFDA:20}\\
				/* Hovering power and Transmit power optimization starts */
				
				\FOR {$j=1$ to $K$} \label{alg:AFDA:21}
				
				\IF{${\rm\bf M}\left(j\right)==1$} \label{alg:AFDA:22}
				
				\STATE Determine the selected suitable UAV-BS's hovering power $P_{j}^{{\rm Hov}}$ from~\eqref{Eq:01}; \label{alg:AFDA:23}
				
				\STATE ${\rm\bf P^{Hov}}\left(j\right)={p_{j}^{\rm Hov}}$; \label{alg:AFDA:24} \hfill // Commit the updated ${p_{j}^{\rm Hov}}$ 
				
				\FOR{$i=1$ to $|\rm\bf\Omega_j|$}	\label{alg:AFDA:25}		
				
				\STATE Update ${\rm\bf D}\left(i,j\right)=\sqrt{{\rm{{\bf R}}}^2\left({i},{j}\right)+{\rm{{\bf H}}}^2\left({j}\right)}$; \label{alg:AFDA:26}
				
				\STATE Determine the average path loss between selected user and UAV-BS, $PL_{h_{j},r_{i,j}}^{{\rm Avg}}$, by~\eqref{Eq:05} with ${\rm\bf D}\left(i,j\right)$; \label{alg:AFDA:27}
				
				\STATE Determine the minimum required transmit power $p_{i,j}^{\min}$ for guaranteeing the SNR value~\eqref{Eq:06} of the $i$-th user by $p_{i,j}^{\min} =10^{{\left(\gamma _{{\rm th}} +PL_{h_{j},r_{i,j} }^{{\rm Avg}}\right)\mathord{ \left/{\vphantom{\left(\gamma_{{\rm th}}+PL_{h_{j},r_{i,j} }^{{\rm Avg}}\right)10}}\right.\kern-\nulldelimiterspace}10}}$; (see appendix-\ref{appendix-A}) \label{alg:AFDA:28}
				
				\STATE Determine the optimal transmit power $p_{i,j}$ to maximize the $E_{j}$ by~\eqref{Eq:12}; \label{alg:AFDA:29}\\
				/* Check the transmit power constraint to make sure $p_{i,j}$ is reasonable, and update $p_{i,j}$. */
				
				\IF{$p_{i,j}^{\min} \leq {p_{\max}}$} \label{alg:AFDA:30}
				
				\IF{$p_{i,j} \leq p_{i,j}^{\min}$} \label{alg:AFDA:31}
				
				\STATE $p_{i,j}=p_{i,j}^{\min}$; \label{alg:AFDA:32}
				
				\ELSIF{$p_{i,j} \geq {p_{\max}}$} \label{alg:AFDA:33}
				
				\STATE $p_{i,j} = {p_{\max}}$; \label{alg:AFDA:34}
				
				\ENDIF \label{alg:AFDA:35}
				
				\ELSE \label{alg:AFDA:36}
				
				\STATE $p_{i,j} = {p_{\max}}$; \label{alg:AFDA:37}
				
				\ENDIF  \label{alg:AFDA:38}			
				
				\STATE ${\rm{{\bf P}}}\left(i,j\right)={p_{i,j}}$; \label{alg:AFDA:39} \hfill // Commit the updated ${p_{i,j}}$

				\ENDFOR \label{alg:AFDA:40}
				\ENDIF \label{alg:AFDA:41}
				\ENDFOR \label{alg:AFDA:42}
				
				\STATE send ${\rm{\bf H}}$, ${\rm\bf P^{Hov}}$ and ${\rm \bf P}$ to all UAV-BSs;
				\label{alg:AFDA:43}
			\end{algorithmic}
		\end{multicols}
	\end{algorithm*}

	\begin{algorithm}[!t]
		\caption{The procedure of re-association} 
		\label{alg:reassociation_function}
		\small
		\textbf{Input:} \\
		\hspace*{1em} Assume that all the variables are shared by the main function (pass-by-reference);\\
		\hspace*{1em} $j$: the overloaded UAV-BS;\\
		\textbf{Output:} \\ 
		\hspace*{1em} the selected user $i^{*}$;\\
		\hspace*{1em} the selected UAV-BS $j^{*}$;\\
		\textbf{Pseudo-code:}
		\begin{algorithmic}[1]		
			\IF{$r\_s==0$} \label{alg:reassociation_function:step1}
			\STATE Find ${i^{*}} = \mathop {{\mathop{\rm argmin}\nolimits} }\limits_{i \in {\rm\bf\Omega _j}} {\gamma _{i,j}}$ 
			\label{alg:reassociation_function:step2} \hfill /* SNR-aware */		
			
			\STATE Calculate $\theta_{j'}^{\rm opt}$ (degree) by~\eqref{Eq:19}, $\forall j'\neq j$;
			\label{alg:reassociation_function:step3}
			
			
			\STATE Calculate $h_{j'} = {\rm\bf R}\left({i^{*},j'}\right)\tan\theta_{j'}^{\rm opt},\forall j'\neq j$;
			\label{alg:reassociation_function:step4}
			
			\STATE Calculate the Euclidean distance $d_{i^{*},j'}$ between user $i^{*}$ and UAV-BS $j'$ by $d_{i^{*},j'}={\rm\bf R}\left({i^{*},j'}\right)\sec\theta_{j'}^{\rm opt},\forall j'\neq j$;
			\label{alg:reassociation_function:step5}
			
			\STATE Find $j^{*} = \mathop {{\mathop{\rm argmin}\nolimits} }\limits_{1 \leq j' \leq K\wedge j'\neq j} d_{i^{*},j'}$;
			\label{alg:reassociation_function:step6}
			
			\ELSIF{$r\_s==1$} \label{alg:reassociation_function:step7}
			
			\STATE Find ${j^{*}} = \mathop {{\mathop{\rm argmin}\nolimits} }\limits_{1\leq j'\leq K\wedge j'\neq j} {|\rm\bf \Omega_{j'}|}$; \label{alg:reassociation_function:step8} \hfill /* Load-aware */
			
			
			\STATE Find $i^{*} = \mathop {{\mathop{\rm argmin}\nolimits} }\limits_{\forall i \in \Omega_{j}} {\rm\bf R}\left({i,j^{*}}\right)$; \label{alg:reassociation_function:step9}
			
			
			\ELSE
			
			\STATE Find $i^{*} = {\rm rand}\left(\rm\bf \Omega_{j}\right)$; \label{alg:reassociation_function:step11} \hfill /* Random user handover */
			
			\STATE Find $j^{*} = \mathop {{\mathop{\rm argmin}\nolimits} }\limits_{1\leq j'\leq K\wedge j'\neq j} {\rm\bf R}\left({i^{*},j'}\right)$;  \label{alg:reassociation_function:step12}
			
			\ENDIF	
			
			\RETURN $i^{*}, j^{*}$;
		\end{algorithmic}
	\end{algorithm}

	\subsection{The procedure of AFD} 
	The required input information comprises $\omega_{\max}$, $h_{\min}$, $h_{\max}$, $r\_s$, $\rm\bf E$, $\rm\bf U$, $\rm\bf\Omega$, and ${\rm \bf H}$.  
	Here, $\omega_{\max}$ is the predefined maximum allowable number of users that can be served by the ${j}$-th UAV-BS; $h_{\min}$ and $h_{\max}$ are the the minimum altitude and the maximum altitude limits of a UAV-BS; $r\_s$ is the input value to select the re-association scheme; ${\rm{\bf E}} = \left\{ {{u_1},\dots,{u_i},\dots ,{u_N}} \right\}$ is the collected horizontal location information of all users; ${\rm{\bf U}} = \left\{ {{U_1},\dots,{U_j},\dots ,{U_K}} \right\}$ is the collected horizontal location information of all users; ${\rm\bf\Omega}=\{{\rm\bf\Omega_1},\dots,{\rm\bf\Omega_K}\}$ is an association map to record all the associations of all UAV-BSs; $\rm\bf\Omega_{j}$ is the set of users associated with the $j$-th UAV-BS; $|{\rm\bf\Omega_{j}}|$ represents the number of user association on the $j$-th UAV-BS; and ${\rm\bf H}=\left\{h_{1},\dots,h_j,\dots,h_K\right\}$ is a matrix to record the altitudes of all UAV-BSs, where $h_{j}$ is the altitude of the ${j}$-th UAV-BS, where $\forall i\in\{1,2,\dots,N\}$ and $\forall j\in\{1,2,\dots,K\}$.
	
	
	Algorithm~\ref{alg:AFDA} presents the pseudo-code of the main procedure of the proposed AFD. Note that this AFD procedure is performed by the CVCC. Each step of the AFD procedure is described in detail as follows:
	\begin{itemize}
		\item From steps~\ref{alg:AFDA:01} to~\ref{alg:AFDA:05} are initial phase of AFD, the CVCC prepares some temporary matrices to record temporary information to help the subsequent re-association decision and optimization.
		\item At the step~\ref{alg:AFDA:06}, the CVCC uses a for-loop to iteratively check whether each UAV-BS is overloaded.
		\item From step~\ref{alg:AFDA:07} to~\ref{alg:AFDA:19}, the CVCC runs the while loop to monitor the number of excess users of the $j$-th UAV-BS. If $\Omega_j^{{\rm Excess}}>0$, the CVCC repeatedly do the re-association and the altitude optimization until all the excess users are re-associated and $\Omega_j^{{\rm Excess}}=0$ is satisfied, which means a load balancing decision is made to solve the overload problem of the $j$-th UAV-BS.
		\item Step~\ref{alg:AFDA:08} is to do the re-association phase in Algorithm~\ref{alg:reassociation_function}. The function is mainly to update the association between the selected excess user and the adjacent available UAV-BS. The detail procedure of this function will be introduce separately latter.	
		\item Step~\ref{alg:AFDA:09} and~\ref{alg:AFDA:10} is to find the optimal altitude of the selected neighboring available UAV-BS, $U_{j^{*}}$, to cover the selected excess user from the overloaded cell of UAV-BS, $U_j$. 
		\item However, the steps from~\ref{alg:AFDA:11} to~\ref{alg:AFDA:15} are to ensure compliance with local laws regarding UAV height restrictions. 
		\item Steps~\ref{alg:AFDA:16} and~\ref{alg:AFDA:17} handover the association from the overloaded UAV-BS, $U_j$, to the selected neighboring available UAV-BS, $U_{j^{*}}$.
		\item After determining the updated altitude value, $h_{j^{*}}$, step~\ref{alg:AFDA:18} commits/save the altitude to the parameter set, $\rm\bf H$.
		\item After the altitude optimization, starts a new for loop from steps~\ref{alg:AFDA:21} to~\ref{alg:AFDA:42} with an temporary array $\rm\bf M$ at step~\ref{alg:AFDA:22} to optimizes the hovering power and transmit power of each UAV-BS. With the help of $\rm\bf M$, the CVCC only needs to optimize the hovering power and transmit power of the $j$-th UAV-BS if ${\rm\bf M}\left(j\right)=1$.
		\item Step~\ref{alg:AFDA:23} uses~\eqref{Eq:01} to determine the optimal hovering power of the $j$-th UAV-BS, $P_j^{\rm Hov}$.
		\item Step~\ref{alg:AFDA:24} commits/save the power value to the parameter set, $\rm\bf P^{Hov}$.
		\item From steps~\ref{alg:AFDA:25} to~\ref{alg:AFDA:40}, the CVCC re-allocates an optimal transmit power of each association link from the $j$-th UAV-BS to the $i$-th user.	
		\item Step~\ref{alg:AFDA:26} computes the Euclidean distance between the $j$-th UAV-BS to the $i$-th user since the altitude of the $j$-th UAV-BS has changed.
		\item Step~\ref{alg:AFDA:27} uses~\eqref{Eq:05} to update the path loss of the association link between the $j$-th UAV-BS to the $i$-th user since the Euclidean distance, ${\rm\bf D}\left(i,j\right)$, also has changed.	
		\item For the link association between the selected $i$-th user and the $j$-th UAV-BS, step~\ref{alg:AFDA:28} uses the given SNR threshold, $\gamma_{\rm th}$~\eqref{Eq:06}, and the updated path loss value, $PL_{h_{j^{*}},r_{i^{*},j^{*}}}^{\rm Avg}$, from step~\ref{alg:AFDA:27} to calculate the minimum required transmit power, $p_{i^{*},j^{*}}^{\min }=10^{{\left(\gamma _{{\rm th}} +PL_{h_{j}{*},r_{i^{*},j^{*}} }^{{\rm Avg}}\right)\mathord{ \left/{\vphantom{\left(\gamma_{{\rm th}}+PL_{h_{j}{*},r_{i^{*},j^{*}} }^{{\rm Avg}}\right)10}}\right.\kern-\nulldelimiterspace}10}}$ (see appendix-\ref{appendix-A}).
		\item Step~\ref{alg:AFDA:29} tries to increase $p_{i,j}$ to maximize the energy efficiency of the $j$-th UAV-BS, $E_j$, by~\eqref{Eq:12}.
		\item Steps~\ref{alg:AFDA:30} to~\ref{alg:AFDA:38} check the minimum and maximum constraints of transmit power make sure that $p_{i,j}$ is reasonable and then commit the updated $p_{i,j}$ to ${\rm\bf P}\left(i,j\right)$ at step~\ref{alg:AFDA:39}.
		\item Finally, the CVCC send the updated parameter sets, $\rm\bf H$, $\rm\bf P^{Hov}$, and $\rm\bf P$, to all UAV-BSs to update the deployment of UAV-BSs.
	\end{itemize}
	
	In addition to the main procedure of AFD, the pseudo-code of the re-association function is shown in Algorithm~\ref{alg:reassociation_function}. Each step of the re-association function is also described in detail as follows:
	\begin{itemize}
		\item From steps~\ref{alg:reassociation_function:step1} to~\ref{alg:reassociation_function:step6}, the CVCC will do the SNR-aware re-association if the given scheme selection parameter, $r\_s$, is 0.
		\item The CVCC using SNR-aware re-association firstly finds a excess user, $i^{*}$, with the minimum SNR from the association set of the overloaded UAV-BS, $\rm\bf\Omega_j$ at step~\ref{alg:reassociation_function:step2}. Note that $i^{*}$ is a pointer to user $u_{i^{*}}$.
		\item Step~\ref{alg:reassociation_function:step3} uses~\eqref{Eq:19} to calculate the optimal elevation angle of the $j'$-th neighboring UAV-BSs, where $\forall j'\neq j$.
		\item At step~\ref{alg:reassociation_function:step4}, to cover each re-associated excess user, $i^{*}$, the $j'$-th neighboring UAV-BS needs provide at least a horizontal distance coverage radius, ${\rm\bf R}(i^{*},j')$. With the horizontal distance coverage radius, ${\rm\bf R}(i^{*},j')$, and the optimal elevation angle of the $j'$-th UAV-BS, $\theta_{j'}^{\rm opt}$, the CVCC can compute the optimal required altitude of the $j'$-th UAV-BS, where $\forall j'\neq j$.
		\item The CVCC can use the optimal elevation angle and altitude of the $j'$-th neighboring UAV-BS to derive the Euclidean distance from the selected user, $i^{*}$, to the $j'$-th neighboring UAV-BSs at step~\ref{alg:reassociation_function:step5}, where $\forall j'\neq j$.
		\item The final step (step~\ref{alg:reassociation_function:step6}) of the SNR-aware scheme selects the neighboring UAV-BS with the smallest Euclidean distance as the new associated UAV-BS with the selected excess user, $i^{*}$.
		\item If given $r\_s=1$, the CVCC will execute the load-aware re-association scheme from steps~\ref{alg:reassociation_function:step7} to~\ref{alg:reassociation_function:step9}.
		\item The CVCC using load-aware re-association scheme first selects a neighboring UAV-BS to the overloaded UAV-BS with the minimum $|\rm\bf\Omega_j|$ at step~\ref{alg:reassociation_function:step8}.
		\item The CVCC then selects the user from the overloaded UAV-BS that is closest to the selected neighboring available UAV-BS at step~\ref{alg:reassociation_function:step9}.
		\item For the random user handover scheme, the CVCC randomly selects a user from the association set, $\rm\bf\Omega_j$ at step~\ref{alg:reassociation_function:step11}.
		\item After that, the CVCC selects the neighboring UAV-BS that is closest to the selected user at step~\ref{alg:reassociation_function:step12}.
		\item At the final step, the CVCC outputs the selected user $i^{*}$ and UAV-BS $j^{*}$.
	\end{itemize}

	\subsection{Complexity Analysis}
	\label{sec:complexity}
	In this section, we firstly discuss the complexity of re-association in Algorithm~\ref{alg:reassociation_function} since it is use in the main procedure of AFD in Algorithm~\ref{alg:AFDA}.
	Using SNR-aware re-association scheme, the CVCC executes the steps from~\ref{alg:reassociation_function:step2} to~\ref{alg:reassociation_function:step6}. The time complexity of SNR-aware re-association scheme is $\mathcal{O}\left(|{\rm\bf\Omega_j}|+4(K-1)\right)$ since only step~\ref{alg:reassociation_function:step2} costs $\mathcal{O}(|{\rm\bf\Omega_j}|)$ time and the following each step costs $\mathcal{O}(K-1)$ time, respectively.
	For load-aware re-association scheme, step~\ref{alg:reassociation_function:step8} takes $\mathcal{O}(K-1)$ time and step~\ref{alg:reassociation_function:step9} takes $\mathcal{O}(|{\rm\bf\Omega_j}|)$ time. Hence, the time complexity of load-aware re-association scheme is $\mathcal{O}(|{\rm\bf\Omega_j}|+k-1)$.
	For the last scheme, random user handover, step~\ref{alg:reassociation_function:step11} only costs $\mathcal{O}(1)$ time and step~\ref{alg:reassociation_function:step12} costs $\mathcal{O}(K-1)$ time. So, the time complexity of random user handover scheme is $\mathcal{O}(K)$.
	
	With two for-loops at steps~\ref{alg:AFDA:06} and~\ref{alg:AFDA:21}, the proposed AFD algorithm can multiple overloaded UAV-BS situations. However, according to our considered system model and assumption, we only discuss single overloaded UAV-BS case. Suppose $n^{\rm Excess}=|{\rm\bf\Omega_j}|-\omega_{\max}$ is the number of excess users, consider the operations from steps~\ref{alg:AFDA:01} to steps~\ref{alg:AFDA:21} of main procedure, the time complexity of using SNR-aware re-association scheme is $\mathcal{O}\left(n^{\rm Excess}\cdot\left(|{\rm\bf\Omega_j}|+4(K-1)\right)\right)$. For load-aware re-association scheme, the time complexity is $\mathcal{O}\left(n^{\rm Excess}\cdot\left(|{\rm\bf\Omega_j}|+K-1\right)\right)$. For random user handover, the time complexity becomes $\mathcal{O}\left(n^{\rm Excess}\cdot K\right)$.
	
	Consider the remaining operations from steps~\ref{alg:AFDA:21} to~\ref{alg:AFDA:43}, this part of procedure is to determine the hovering power of $K$ UAV-BSs and then allocate the required transmit power with respect to all $N$ users. The time complexity will be $\mathcal{O}\left(\sum_{j=1}^{K}|{\rm\bf\Omega_j}|+K\right)=\mathcal{O}\left(N+K\right)$. Hence, the total time complexities of AFD framework using different re-association schemes can be summarized as follows:
	\begin{itemize}
		\item SNR-aware re-association:\\ $\mathcal{O}\left(n^{\rm Excess}\cdot\left(|{\rm\bf\Omega_j}|+4(K-1)\right)+N+K\right)$.
		\item Load-aware re-association:\\ $\mathcal{O}\left(n^{\rm Excess}\cdot\left(|{\rm\bf\Omega_j}|+K-1\right)+N+K\right)$.
		\item Random user handover: $\mathcal{O}\left(n^{\rm Excess}\cdot K+N+K\right)$.
	\end{itemize}

\subsection{Design Discussion}
\label{Design Discussion}
In this section, we summarize the key benefits of the proposed AFD approach as follows:
\begin{itemize}
	\item \textbf{Seamless User Handover.}
	\label{Seamless User Association}
	In the proposed AFD approach, we only allow neighboring UAV-BSs to increase their altitudes and transmit power so that the central UAV-BS-$j$ (where $j=1$) to fly at a relatively low altitude during the traffic offloading process. Such a design makes this easy for users to receive stronger signals from the neighboring UAV-BSs. Then, the users will hand their association seamlessly over to neighboring UAV-BS.
	\item \textbf{No Coverage Outage.} 
	\label{No Coverage Outage}
	The existing approach~\cite{R7} offloads the traffic from overloaded UAV-BS by moving closer to the UAV-BS and sharing the load. When UAV-BS moves toward the overloaded UAV-BS, some or all their users face coverage outages until another UAV-BS comes to serve. Meanwhile, in our approach, the neighboring UAV-BS only changes their altitude, not the location, during traffic offloading; thus, there is no chance of coverage outage.
	\item \textbf{No Additional Hardware Required.} 
	\label{No Additional Hardware Required}
	The proposed AFD approach is simple and easy to implement. The UAV-BS does not need extra hardware to adjust its coverage during traffic offloading from the central/overloaded UAV-BS compared to traditional GBS, where extra hardware is required for antenna tilting~\cite{R24}\cite{5621970}.  
\end{itemize}


\section{Simulation Results and Performance Analysis}
\label{sec: Simulation Results and Performance Analysis}
This section will evaluate the problem~\eqref{Eq:obj_P1} under several performance criteria. We randomly generate users in the central UAV-BS with different numbers of excess users, conduct $1,000$ times Monte Carlo simulations to verify the average performance of the proposed AFD approach. We also compare the proposed schemes, SNR-aware, and load-aware, with two conventional schemes, random user handover and straightforward greedy. 
We assume that all UAV-BSs share the common spectrum and provide equal bandwidth for the downlink transmission to users in the considered system model. 

In our simulations, we consider the urban scenario and the corresponding values of the environmental parameters, $(a,b,\eta_{\rm LoS},\eta_{\rm NLoS})=(9.61,0.16,1,20)$, taken from~\cite{7510820} are initially demonstrated in~\cite{R5} and~\cite{7037248}. 
In order to obtain performance results close to the actual situation, as a reference for future research, we use the allowable height range, $[h_{\min},h_{\max}]=[30,400]$, stipulated by the laws of Taiwan~\cite{laws_taiwan}. 
We assume that the used spectrum frequency is $f_c$, the maximum transmit power of a UAV-BS is $p_{\max}$, the total number of active users is $N$, and the maximum number of user associations per UAV-BS is $\omega_{\rm max}$. Every neighboring UAV-BS cannot accept more than the maximum number of user associations at a time.

The computer simulations are implemented in MATLAB R2020b, and the program observes the performance of comparison approaches per second. Table~\ref{Technical/Physical properties} presents the technical and physical properties of typical UAVs. Table~\ref{Simulation Parameter} presents the numerical parameters used in simulations.

\begin{table}[!t]
	\centering
	\caption{Simulation Parameters}
	\begin{tabular}{|p{3cm}|l|l|}
		\hline
		\textbf{Parameter}&\textbf{Symbol}&\textbf{Value}
		\label{Simulation Parameter}\\\hline 
		Environmental parameters~\cite{7510820} & $(a,b,\eta_{\rm LoS},\eta_{\rm NLoS})$ & $(9.61,0.16,1,20)$ \\ \hline 
		Career frequency & $f_c$ & 2.4 Ghz \\ \hline 
		Speed of light & $c$ & $3*10^8$ \text{m/s} \\ \hline 
		Minimum altitude & $h_{\rm min}$ & 30 m \\ \hline 
		Maximum altitude & $h_{\rm max}$ & 400 m \\ \hline
		Allocated bandwidth & $B_{i,j}$ & 20 MHz \\ \hline 
		Noise power spectral density &$\sigma^2$ & -174 dBm/Hz \\\hline 
		Total Number of UEs & $N$ & 250 \\ \hline
		Number of UAV-BSs & $K$ & 7 \\ \hline
		Maximum number of user associations per UAV-BS & $\omega_{\rm max}$ & $50$ \\ \hline 
		SNR Threshold & $\gamma_{\rm th}$ & 3 dB  \\ \hline
		Maximum transmission power & $p_{\rm max}$ & 29 dBm  \\ \hline
	\end{tabular}
\end{table}

\subsection{Hovering and Total Power consumption}
First, the required power consumption for hovering UAV-BS is plotted as a function of the number of excess users in the central UAV-BS in Fig.~\ref{fig:fig5.eps} for four different schemes, SNR-aware, load-aware, random user handover, and straightforward greedy. The number of excess users and altitude play an essential role in the hovering power consumption of UAV-BS. For a UAV-BS with a certain number of excess users, the power required for hovering increases with the altitude. For instance, we can observe from Fig.~\ref{fig:fig5.eps} that the random user handover consumed 208 mW, and SNR-aware and load-aware consumed 91 mW more power than straightforward greedy. As shown in Fig.~\ref{fig:fig5.eps}, the UAV increases hover power consumption as the number of excess users increases in the central UAV-BS.  

Second, the required total power consumption is plotted as a function of the number of excess users in the central UAV-BS in Fig.~\ref{fig:fig6.eps} for SNR-aware, load-aware, random user handover, and straightforward greedy schemes. The number of excess users plays an essential role in the total power consumption of UAV-BS. We can see that when the number of users increases in the central UAV-BS, the total power consumption decreases significantly in the straightforward greedy scheme yet increases in the proposed SNR-aware and load-aware by 218 mW and random user handover by 333 mW. In contrast, the proposed schemes achieve almost the same total power consumption because excess users offload fairly among neighboring UAV-BS. We will discuss whether the random user handover scheme shows the worst performance in section~\ref{Comparison_Summary}.

\begin{figure}[!t]
	\centering
	\includegraphics[width=.95\columnwidth]{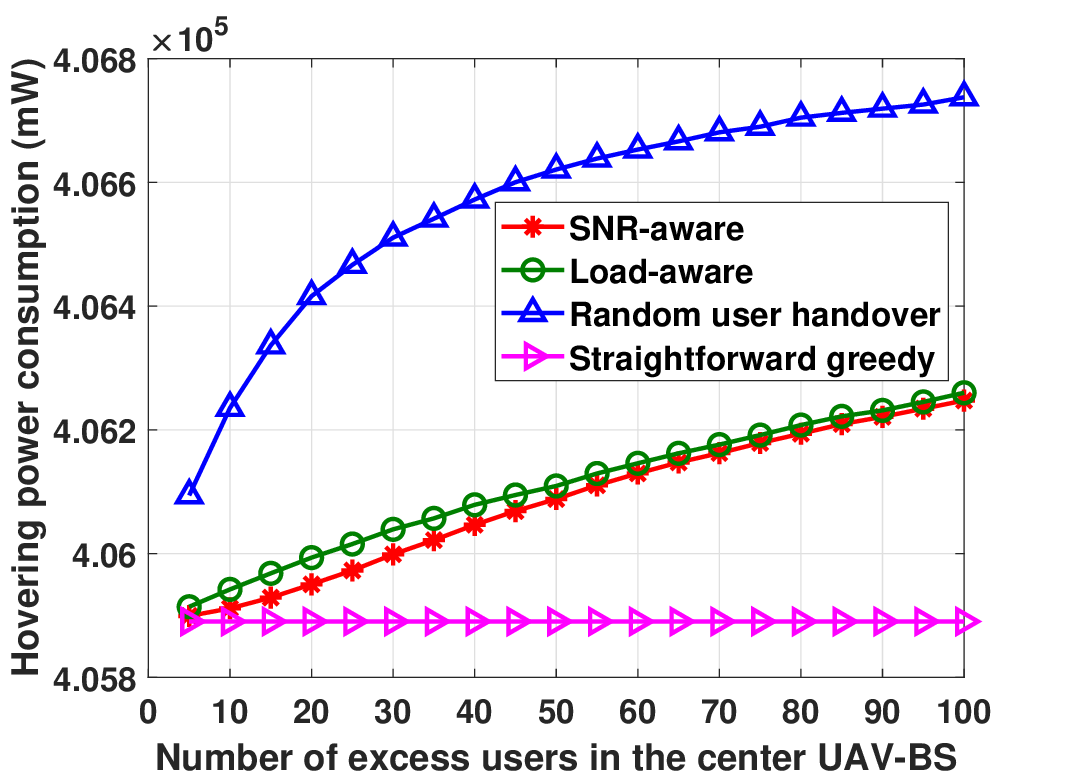}
	\caption{Hovering power consumption (mW) vs Excess users}
	\label{fig:fig5.eps}
\end{figure}

\begin{figure}[!t]
	\centering
	\includegraphics[width=.95\columnwidth]{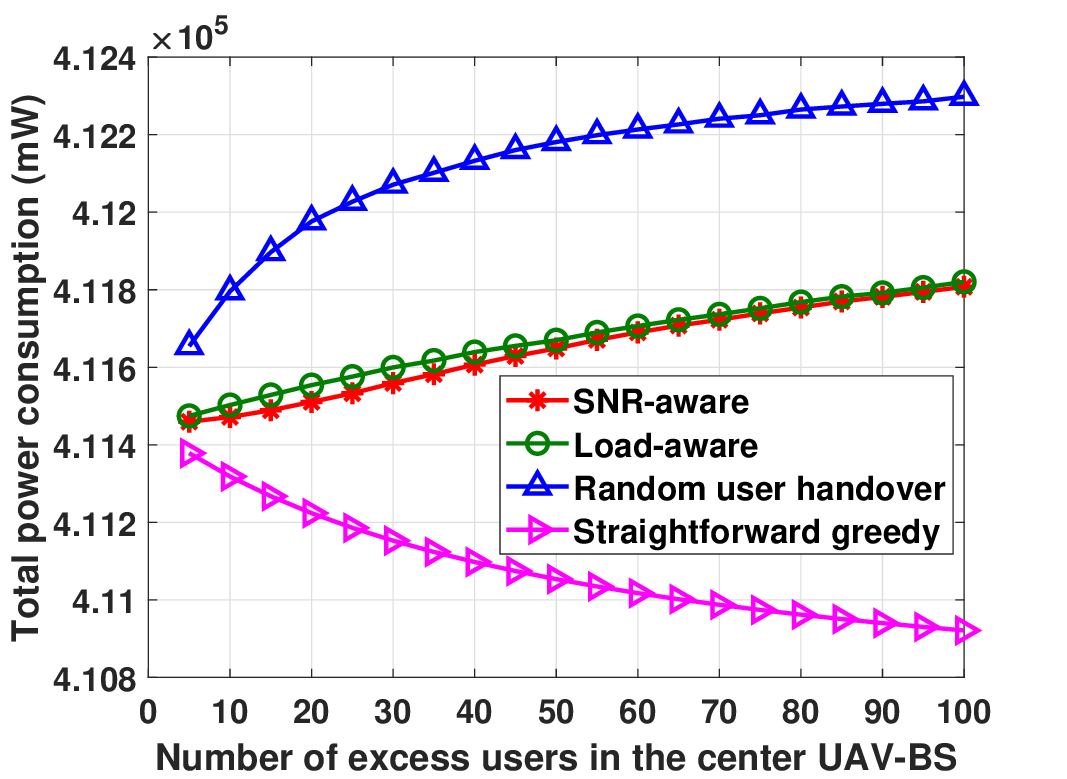}
	\caption{Total power consumption (mW) vs Excess users}
	\label{fig:fig6.eps}
\end{figure}

\subsection{Total Capacity and Total Energy Efficiency}
The numerical result in Fig.~\ref{fig:fig7.eps} shows that the proposed schemes improve the system's total capacity more than random user handover and straightforward greedy schemes. Compared with the straightforward greedy scheme, the proposed SNR-aware and load-aware approach improved total capacity by 37.71\% to 363 Mbps, whereas the random user handover approach improved by 36.09\% to 226 Mbps. 

\begin{figure}[!t]
	\centering
	\includegraphics[width=.95\columnwidth]{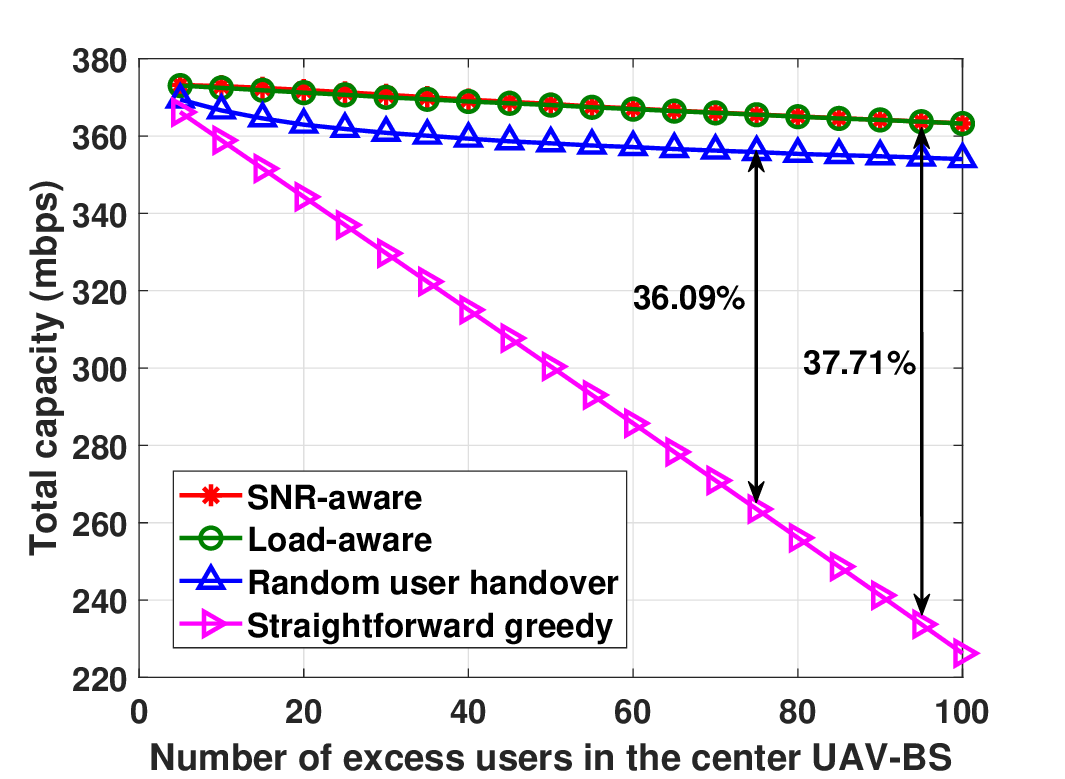}
	\caption{Total capacity (mbps) vs Excess users}
	\label{fig:fig7.eps}
\end{figure}

Fig.~\ref{fig:fig8.eps} hows the UAV-BS network operator with the statistics of how the total energy efficiency is affected by the change in excess user density in the central UAV-BS. Our goal is to propose an AFD approach for traffic offloading in UAV-BS networks and use the available resources to their maximum potential and conserve energy efficiency. In Fig.~\ref{fig:fig8.eps}, we consider SNR-aware, load-aware, random user handover, and straightforward greedy schemes. As the excess user density increases, the total energy efficiency decreases. As the resources are not fairly distributed, it means straightforward greedy schemes with neighboring UAV-BS. Therefore, the excess users in the cell come at the cost of degrading the total energy efficiency. Meanwhile, the proposed scheme better utilized the resources than the straightforward greedy scheme. Thus, the overall total energy efficiency improved by approximately 37.48\% compared with 35.79\% of the random user handover scheme. 

\begin{figure}[!t]
	\centering
	\includegraphics[width=.95\columnwidth]{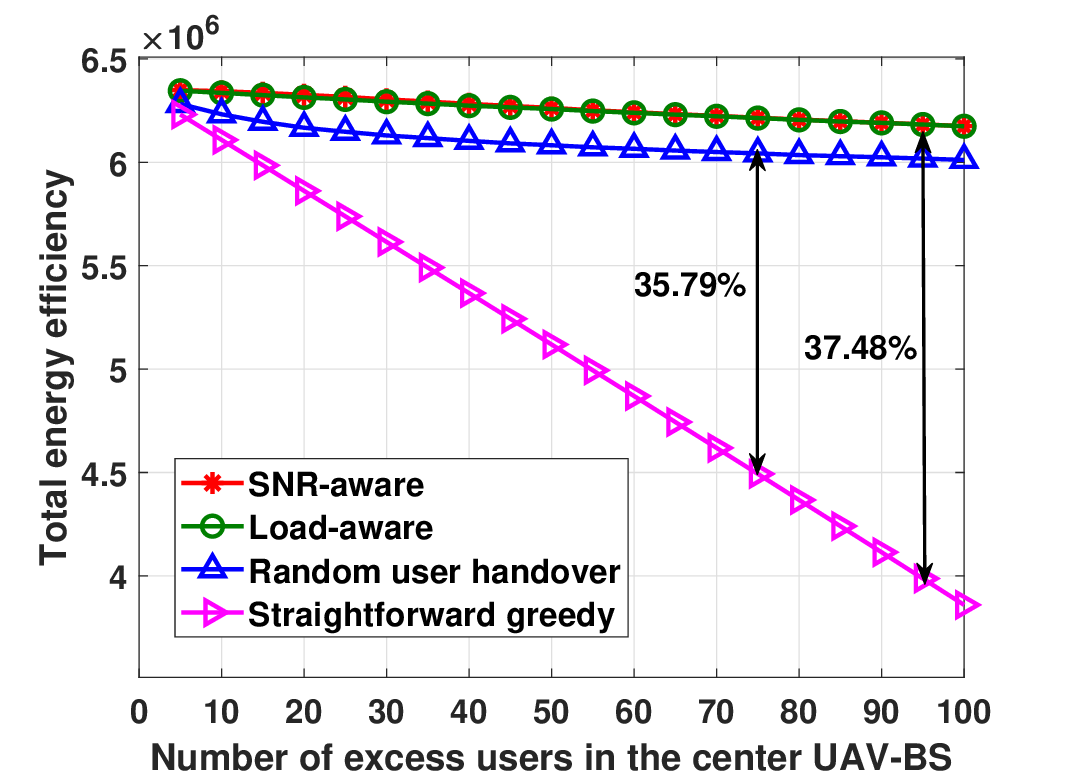}
	\caption{Total energy efficiency vs Excess users}
	\label{fig:fig8.eps}
\end{figure}

\subsection{Jain Fairness Index Value}
This section presents the fairness index of UAV-BSs' load in $1000$ Monte-Carlo iterations. The fairness index examines the excess user distribution among UAV-BSs, defined by JFI~\cite{R19}. In each run, 250 users' locations are randomly generated under the target area covered by seven UAV-BSs. Fig.~\ref{fig:fig9.eps} shows the simulation result of the fairness index. This figure shows that the fairness index in the load-aware Fig.~\ref{fig:implementation:2} is larger than that of the others schemes by 16.12\% compared with random user handover and SNR-aware of 14.74\%. 

\begin{figure}[!t]
	\centering
	\includegraphics[width=.95\columnwidth]{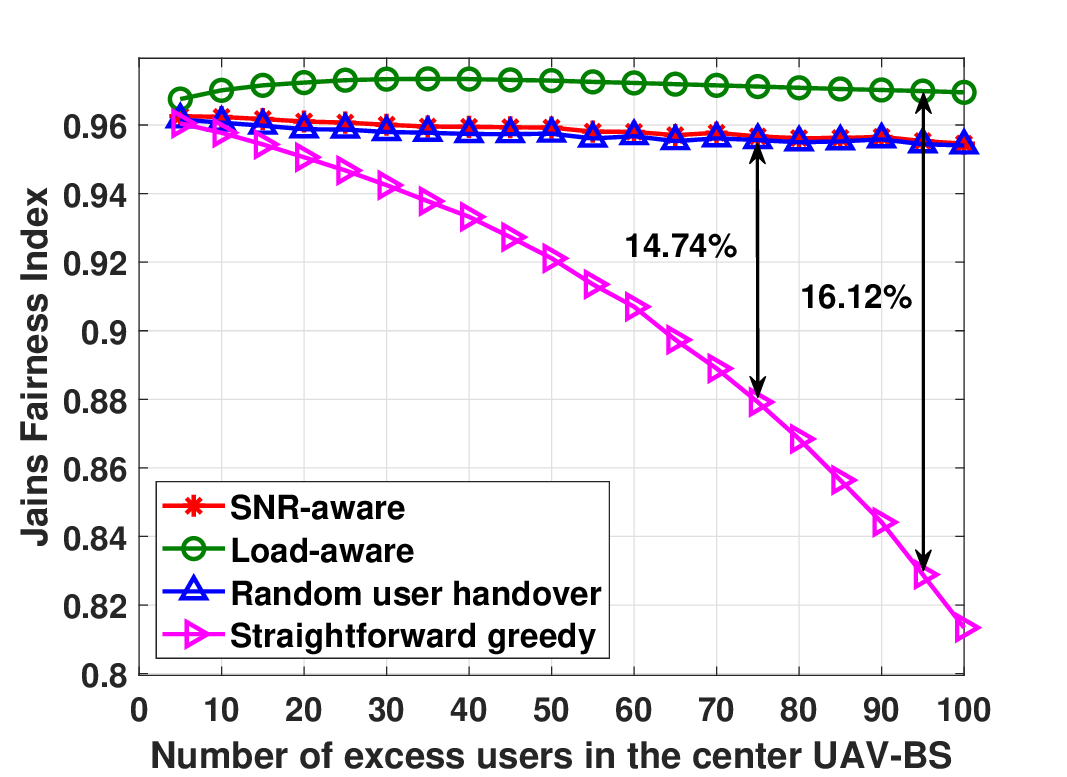}
	\caption{Jain Fairness Index vs Excess users}
	\label{fig:fig9.eps}
\end{figure}

\subsection{Comparison Summary}
\label{Comparison_Summary}
This section presents the comparison results of four schemes with six different simulation parameters shown in Table~\ref{Summary_of_Comparisons}. We consider the comparison characteristics in three parts: low, medium, and high.

For both metrics, hovering power and total power consumption, the SNR-aware and load-aware schemes achieve a medium value compared to the random user handover and straightforward greedy schemes with high and low values. The SNR-aware and load-aware schemes use only two neighboring UAV-BSs to maximize their altitude and transmit power during the traffic offloading. The random user handover scheme used three neighboring UAV-BSs. By using extra UAV-BS~(Fig.\ref{fig:implementation}) for traffic offloading, random user handover consumed more power (hover and transmit) than the SNR-aware and load-aware schemes.

The SNR-aware and load-aware schemes show a higher value than random user handover and straightforward greedy schemes with medium and low values for both metrics' total capacity and energy efficiency. As shown in Figs.~\ref{fig:fig8.eps} and~\ref{fig:fig9.eps}, as the number of excess users increases, the network performance improves because the proposed schemes offload traffic equally with less UAV-BS involvement.

The SNR-aware and random user handover schemes achieve a medium JFI value, whereas the load-aware and straightforward greedy schemes achieve high and low values. The proposed load-aware approach achieves a better JFI value because of fair user distribution at each UAV-BS than SNR-aware and random user handover schemes. In summary, this is the trade-off between traffic offloading and power consumption (see Figs.~\ref{fig:fig5.eps} and~\ref{fig:fig6.eps}).  

Finally, we also summarize the time complexity of different schemes. We analyze the time complexity of all schemes in detail in Section~\ref{sec:complexity}, the time complexity of the SNR-aware scheme is the worst, the load-aware scheme is slightly better, and the random user handover scheme is the best. However, they do not differ very much in time complexity and can all be simply classified as polynomial time algorithms.

\begin{table*}[!t]
	
	\centering
	\caption{Comparative Summary of Simulation Performance}
	\label{Summary_of_Comparisons}
	\begin{tabular}{|l|llll|}
		\hline
		\multirow{2}{*}{Metric} & \multicolumn{4}{c}{Method} \vline \\ \cline{2-5}
		& SNR-aware & \textbf{Load-aware} & Random User Handover & Straightforward greedy \\ \hline
		Hovering Power          & Medium    & \textbf{Medium}     & High           & Low       \\ 
		Total Power             & Medium    & \textbf{Medium}     & High           & Low       \\ 
		Total Capacity          & High      & \textbf{High}       & Medium         & Low       \\ 
		Total Energy   Efficiency & High      & \textbf{High}       & Medium         & Low       \\ 
		Jain's Fairness Index                    & Medium    & \textbf{High}       & Medium         & Low       \\ 	
		Time Complexity  & Polynomial time & \textbf{Polynomial time} & Polynomial time & -       \\ \hline
	\end{tabular}
\end{table*}









\section{Conclusion}
\label{sec:conclusion}
In this work, we proposed a novel 3D AFD algorithm for a multi-UAV-BS network. The proposed AFD can automatically adjust the altitude and transmit power of the UAV-BS. We also identify a new UAV-BS overload problem, which may occur in Multi-UAV-BS networks when users are mobile and unevenly distributed. To solve this problem, we propose an optimization problem to maximize the total energy efficiency and the total capacity of the multi-UAV-BS network by jointly optimizing the altitude and transmission power of the UAV-BS. The simulation result shows that the proposed AFD can improve the total capacity by 37.71\%, total energy efficiency by 37.48\%, and better fairness index by 16.12\% value compared to random user handover and straightforward greedy schemes.



\appendices
\section{Proof of The Minimum Required Transmit Power} 
\label{appendix-A} 
According to~\eqref{Eq:06}, the guaranteed signal transmission the user received SNR, $\gamma_{i,j}$, should always be greater than or equal to the SNR threshold, $\gamma_{\rm th}$. Thus, the received power $p_{\rm r}$ (mW) must be greater than or equal to a corresponding power required, which is
\begin{align}
	{{\mkern1mu\raise4pt\hbox{.}\mkern1mu\raise1pt\hbox{.}\mkern1mu\raise4pt\hbox{.}\mkern1mu}p_{r}\ge 10^{\frac{\gamma_{{\rm th}}}{10}}}.
	\label{appendix-A:Eq:21}
\end{align}
With~\eqref{appendix-A:Eq:21}, we can get the minimum required transmission power, $p_{i*,j*}^{\min}$, by 
\begin{align} 
	\Rightarrow & p_{r}=p_{i*,j*}\cdot 10^{-\frac {PL_{h_{j}*,r_{i*,j*}}^{{\rm Avg}}}{10}} \ge 10^{\frac{\gamma_{{\rm th}}}{10}},    
	\label{appendix-A:Eq:22}\\
	\Rightarrow & p_{i*,j*}^{\min}=10^{\frac{\gamma_{{\rm th}}}{10}}10^{\frac {PL_{h_{j}*, r_{i*,j*}}^{{\rm Avg}}}{10}}=10^{\frac{\gamma_{{\rm th}}+PL_{h_{j}{*},r_{i*,j*}}^{{\rm Avg}}}{10}}.   
	\label{appendix-A:Eq:23}
\end{align}

\section{Proof of~\eqref{Eq:13}}
\label{appendix-B}   
By expanding~\eqref{Eq:12}and substituting the SNR value from~\eqref{Eq:06}, we obtain
\begin{align}  
	E=\frac{B_{i,j}\log_{2}\left(1+\frac{p_{i,j} \cdot PL_{{h_j},{r_{i,j}}}^{{\rm{Avg}}}} {B_{i,j}\sigma^2}\right)}{p_{i,j} +P_{j}^\text{Hov}}. 
	\label{appendix-B:Eq:24}
\end{align}
To find the transmission power, we use the partial differentiation $\frac{\partial E_{j}}{\partial p_{i,j} }=0$ of nonlinear equation~\eqref{appendix-B:Eq:24} to obtain the maximum energy efficiency. Thus, we have

\begin{align} 
	\vspace{-1em}
	\Rightarrow\frac{\partial E}{\partial p_{i,j}}&=\left(B_{i,j}\log_{2}\left(1+\frac{p_{i,j} \cdot PL_{{h_j},{r_{i,j}}}^{{\rm{Avg}}}}{B_{i,j}\sigma^2}\right)\right)^{{'}}\left(p_{i,j}+P_{j}^\text{Hov}\right)^{-1}\nonumber\\
	&+\left(B_{i,j}\log_{2}\left(1+\frac{p_{i,j} \cdot PL_{{h_j},{r_{i,j}}}^{{\rm{Avg}}}} {B_{i,j}\sigma^2}\right)\right) \left(\left(p_{i,j}+P_{j}^\text{Hov}\right)^{-1}\right)^{{'}} 
	\label{appendix-B:Eq:25}\\
	&=B_{i,j}\frac{\left(1+\frac{p_{i,j}\cdot PL_{ {h_j},{r_{i,j}}}^{{\rm{Avg}}}}{B_{i,j} \sigma^2}\right)^{{'}}}{\left(1+\frac{p_{i,j}\cdot PL_{{h_j},{r_{i,j}}}^{{\rm{Avg}}}} {B_{i,j}\sigma^2}\right)\ln2} \left(p_{i,j}+P_{j}^\text{Hov}\right)^{-1}\nonumber\\
	&+\left(B_{i,j}\log_{2}\left(1+\frac{p_{i,j} \cdot PL_{{h_j}, {r_{i,j}} }^{{\rm {Avg}}}}{B_{i,j}\sigma^2}\right)\right)\nonumber\\
	&\times\left(-\left(p_{i,j}+P_{j}^\text{Hov}\right)^{-2}\left(p_{i,j}+P_{j}^\text{Hov}\right)^{{'}}\right)
	\label{appendix-B:Eq:26}\\
	&=B_{i,j}\frac{\frac{PL_{{h_j},{r_{i,j}}}^{{\rm{Avg}}}}{B_{i,j}\sigma^2}}{\left(1+\frac{p_{i,j}\cdot PL_{{h_j},{r_{i,j}}}^ {{\rm{Avg}}}}{B_{i,j}\sigma^2}\right)\ln2}\left(p_{i,j} +P_{j} ^\text{Hov} \right)^{-1}\nonumber\\
	&+\left(B_{i,j}\log_{2}\left(1+\frac{p_{i,j}\cdot PL_{{h_j},{r_{i,j}}}^{{\rm{Avg}}}}{B_{i,j}\sigma^2}\right)\right)\left(-\left(p_{i,j}+P_{j}^\text{Hov}\right)^{-2}\right) 
	\label{appendix-B:Eq:27}\\
	&=\frac{PL_{{h_j},{r_{i,j}}}^{{\rm{Avg}}}}{\left(1+\frac{p_{i,j}\cdot PL_{i,j}}{B_{i,j} \sigma^2}\right)\left(p_{i,j}+P_{j}^\text{Hov}\right)\sigma^2\ln2}\nonumber\\
	&-\frac{B_{i,j}\log_{2}\left(1+\frac{p_{i,j}\cdot PL_{{h_j},{r_{i,j}}}^{{\rm{Avg}}}} {B_{i,j} \sigma^2}\right)}{\left(p_{i,j}+P_{j}^\text{Hov} \right)^{2}}
	\label{appendix-B:Eq:28}
\end{align} 
\begin{align} 
	\frac{\partial E}{\partial p_{i,j}}=0\Rightarrow&\frac {PL_{{h_j},{r_{i,j}}}^{{\rm{Avg}}}} {\left(1+\frac{p_{i,j}\cdot PL_{{h_j},{r_{i,j}}}^{{\rm{Avg}}}}{B_{i,j}\sigma^2}\right) \left(p_{i,j}+P_{j}^\text{Hov}\right)\sigma^2\ln2}\nonumber\\
	&-\frac{B_{i,j}\log_{2}\left(1+\frac{p_{i,j}\cdot PL_{{h_j},{r_{i,j}}}^{{\rm{Avg} }} }{B_{i,j}\sigma^2}\right)} {\left(p_{i,j}+P_{j}^\text{Hov}\right)^{2}}=0 
	\label{appendix-B:Eq:29}
\end{align} 


\bibliographystyle{IEEEtran}
\bibliography{Reference}
\ifCLASSOPTIONcaptionsoff  \newpage \fi

\begin{IEEEbiography}[{\includegraphics[width=1in,height=1.25in,clip,keepaspectratio]{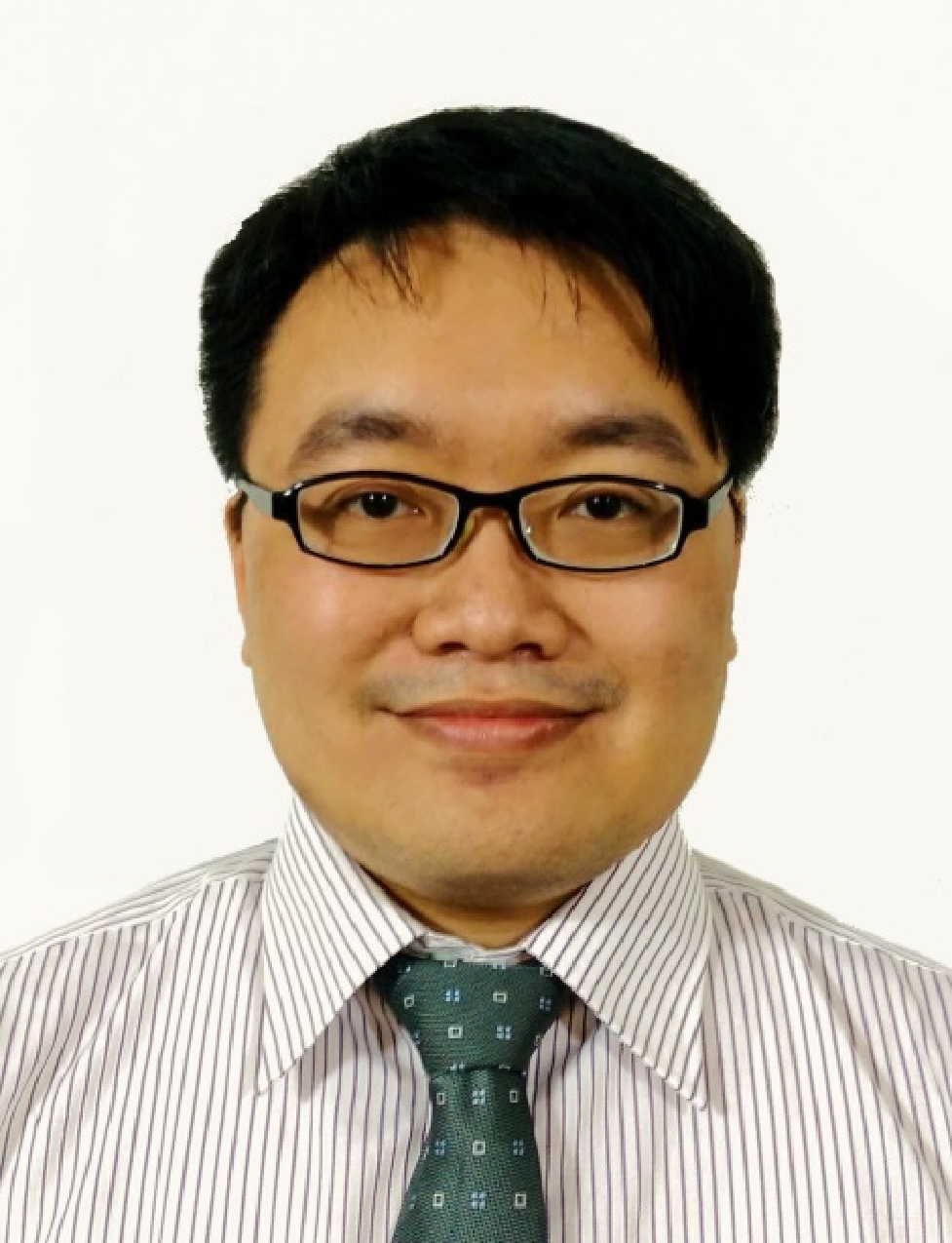}}]{Chuan-Chi Lai}
	(S'13 -- M'18) received the Ph.D. degree in computer science and information engineering from the National Taipei University of Technology, Taipei, Taiwan, in 2017. He was postdoctoral research fellow (from 2017 to 2020) and contract assistant research fellow (2020) with Department of Electrical and Computer Engineering of National Chiao Tung University, Hsinchu, Taiwan.
	He is currently an assistant professor with the Department of Information Engineering and Computer Science, Feng Chia University, Taichung, Taiwan. His current research interests include resource allocation, data management, information dissemination techniques, and distributed query processing over moving objects in emerging applications such as the Internet of Things, edge computing, aerial and mobile wireless applications.	
	Dr. Lai has received the Postdoctoral Researcher Academic Research Award of Ministry of Science and Technology, Taiwan, in 2019, the Best Paper Awards in WOCC 2021 and WOCC 2018 conferences, and the Excellent Paper Award in ICUFN 2015 conference.
\end{IEEEbiography}

\begin{IEEEbiography}[{\includegraphics[width=1in,height=1.25in,clip,keepaspectratio]{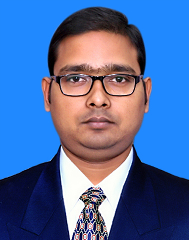}}]{Bhola}
	(S'22) received his B. Tech. Degree in Electronics and Communication Engineering from Dr. A.P.J. Abdul Kalam Technical University, formerly Uttar Pradesh Technical University, Lucknow, Uttar Pradesh, India, in 2011 and M. Tech. degree in Wireless Communication and Networks from School of Information and Communication Technology, Gautam Buddha University, Greater Noida, Delhi NCR, India, in 2015. He is currently pursuing a Ph.D. degree in electrical engineering and computer science-international graduate program (EECS-IGP), National Yang Ming Chiao Tung University, Hsinchu, Taiwan. His current research interests include hand over and mobility management in 5G and beyond, UAV base station deployments, and coverage outage issues. He has been awarded the National Overseas Scholarship by the Government of India in 2017, and an outstanding research award in 2021 by EECS-IGP, National Yang Ming Chiao Tung University, Hsinchu, Taiwan.
\end{IEEEbiography}

\begin{IEEEbiography}[{\includegraphics[width=1in,height=1.25in,clip,keepaspectratio]{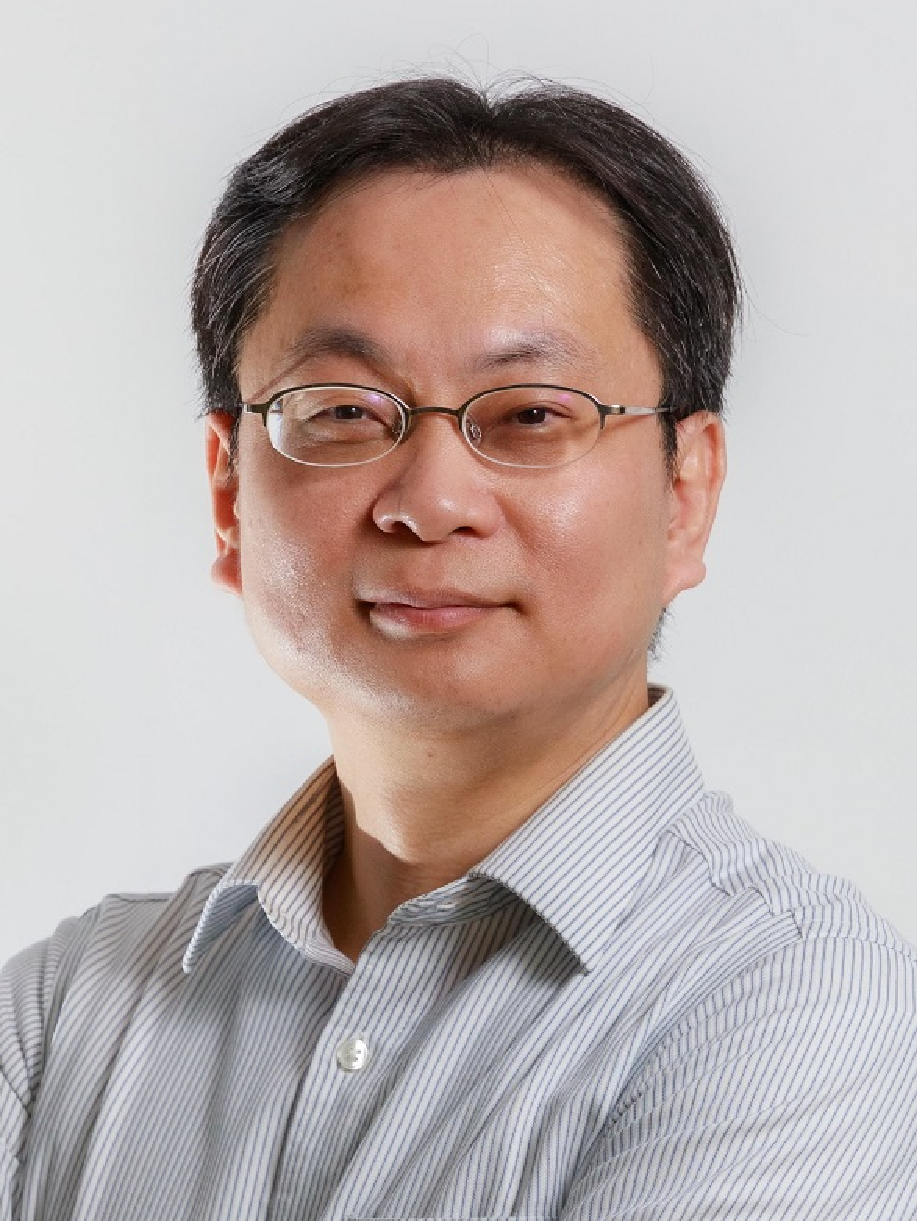}}]{Ang-Hsun Tsai}
	(S'09--M'12) received the Ph.D. degree in communication engineering from the National Chiao-Tung University, Hsinchu, Taiwan, in 2012. He is currently an assistant professor of the Department of Communications Engineering, Feng Chia University in Taiwan. His current research interests include radio resource management in heterogeneous networks, such as 6G mobile networks, non-terrestrial networks, aerial communication networks, and disaster-resilient communication networks.
\end{IEEEbiography}

\begin{IEEEbiography}[{\includegraphics[width=1in,height=1.25in,clip,keepaspectratio]{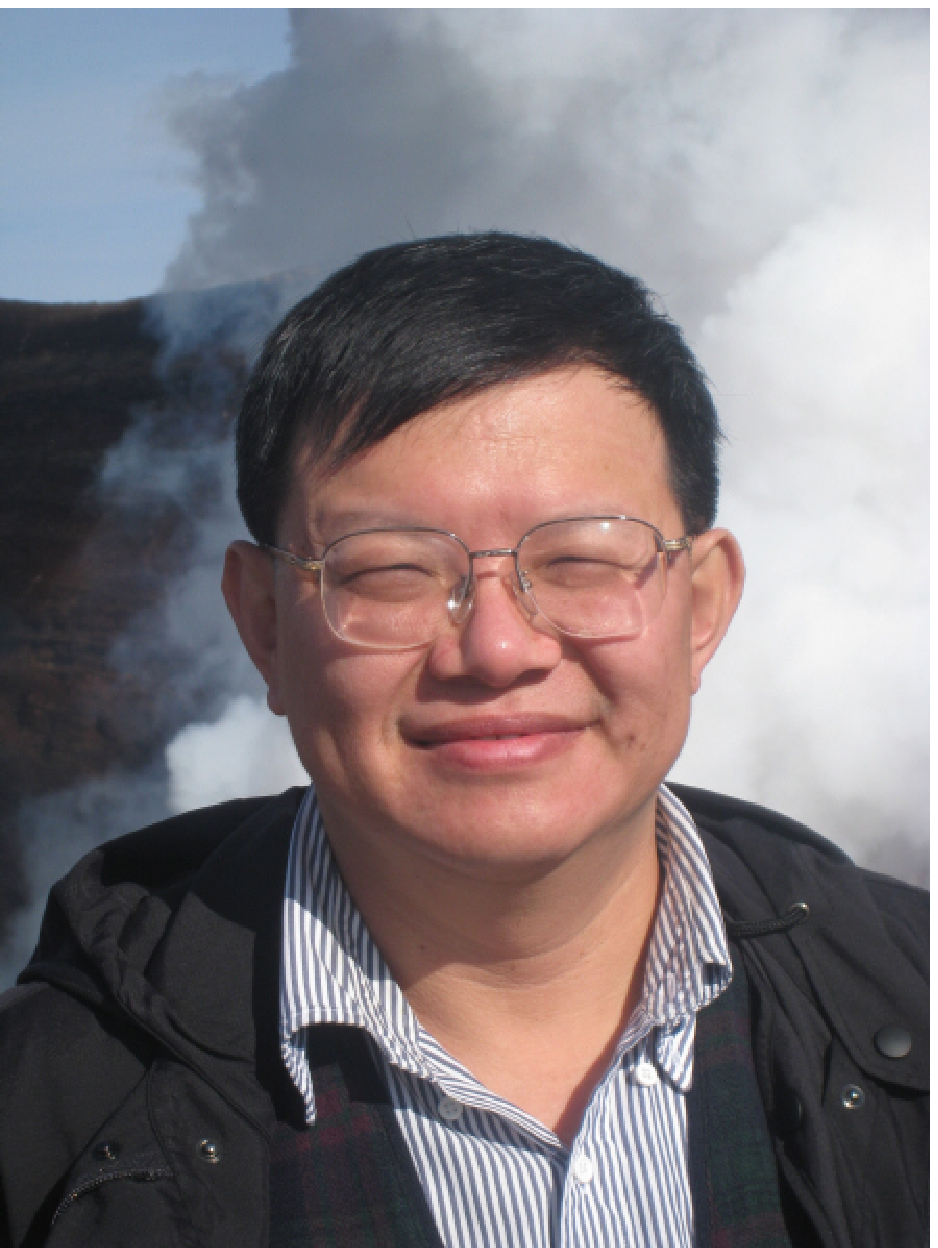}}]{Li-Chun Wang}
	(M'96 -- SM'06 -- F'11) received the Ph.D. degree from the Georgia Institute of Technology, Atlanta, in 1996.
	From 1996 to 2000, he worked with AT\&T	Laboratories, where he was a Senior Technical Staff Member with the Wireless Communications Research Department. In August 2000, he joined the Department of Electrical and Computer Engineering, National Yang Ming Chiao Tung University, Taiwan, where he is currently a Chair Professor and jointly	appointed by the Department of Computer Science and Information Engineering. He holds 26 U.S. patents, and has published over 300 journal and conference papers, and co-edited the book, ``Key Technologies for 5G Wireless Systems" (Cambridge University Press 2017). His recent research interests are in the areas of cross-layer optimization for wireless systems, data-driven radio resource management, software-defined heterogeneous mobile networks, big data analysis for industrial Internet of Things, and AI-enabled unmanned aerial vehicular networks.
	Dr. Wang has won two Distinguished Research Awards from Taiwan’s Ministry of Science and Technology in 2012 and 2017. He was the co-recipients of IEEE Communications Society Asia-Pacific Board Best Award in 2015, the Y. Z. Hsu Scientific Paper Award in 2013, and the IEEE Jack Neubauer Best Paper Award in 1997. He was elected to an IEEE Fellow in 2011, for his contributions to cellular architecture and radio resource management in wireless networks.
\end{IEEEbiography}

\end{document}